\documentclass[sigconf,screen,nonacm]{acmart}
\usepackage[T1]{fontenc}
\usepackage[utf8]{inputenc}

\usepackage{listings, fancyvrb}
\usepackage{multicol}
\usepackage{nicefrac}       
\usepackage{graphicx}
\usepackage{subcaption}
\usepackage[plain]{algorithm}
\usepackage[noend]{algpseudocode}
\usepackage{xspace}
\usepackage{enumitem,kantlipsum}

\usepackage[framemethod=TikZ]{mdframed}
\mdfdefinestyle{listingstyle}{%
  outerlinewidth=0.25pt,outerlinecolor=black,%
  innerleftmargin=5pt,innerrightmargin=5pt,innertopmargin=0pt,innerbottommargin=0pt%
}

\AtBeginDocument{%
	\providecommand\BibTeX{{%
			\normalfont B\kern-0.5em{\scshape i\kern-0.25em b}\kern-0.8em\TeX}}}

\newcommand{\simplock}{\textsc{SimpLock}\xspace}
\newcommand{\indirect}{\textsc{Indirect}\xspace}
\newcommand{\seqlock}{\textsc{SeqLock}\xspace}

\newcommand{\cawf}{{\sc CachedWaitFree}\xspace}
\newcommand{\came}{{\sc CachedMemEfficient}\xspace}

\newcommand{\ourHashtable}{ParlayHash\xspace}
\newcommand{\parlayhash}{ParlayHash\xspace}

\newcommand{\myparagraph}[1]{\textbf{#1.}}

\newcommand{\seqnumber}{seqnum}
\newcommand{\SC}{\textsc{SC}}
\newcommand{\LLSC}{\textsc{LL/SC}}
\newcommand{\LL}{\textsc{LL}}
\newcommand{\VL}{\textsc{VL}}

\newcommand{\future}[1]{} 

\makeatletter
\lst@Key{countblanklines}{true}[t]%
{\lstKV@SetIf{#1}\lst@ifcountblanklines}

\lst@AddToHook{OnEmptyLine}{%
	\lst@ifnumberblanklines\else%
	\lst@ifcountblanklines\else%
	\advance\c@lstnumber-\@ne\relax%
	\fi%
	\fi}
\lst@AddToHook{OnEmptyLine}{\vspace{-0.3\baselineskip}}
\makeatother

\makeatletter
\lst@Key{countblanklines}{true}[t]%
{\lstKV@SetIf{#1}\lst@ifcountblanklines}

\makeatletter
\lst@CCPutMacro
    \lst@ProcessOther {"2A}{%
      \lst@ttfamily 
         {\raisebox{1pt}{*}}
         \textasteriskcentered}
    \@empty\z@\@empty
\makeatother

\newtheorem{theorem}{Theorem}[section]

\newcommand{\true}[0]{\textit{true}\xspace}
\newcommand{\false}[0]{\textit{false}\xspace}

\newcounter{int}
\newcommand{\tab}[1][1]{\setcounter{int}{0}\loop\hspace{\algorithmicindent}\addtocounter{int}{1}\ifnum\value{int}<#1\repeat}

\renewcommand{\phi}{\varphi}
\renewcommand{\theta}{\vartheta}

\renewcommand{\epsilon}{\varepsilon}

\newcommand{\isetsep}{\;\ifnum\currentgrouptype=16 \middle\fi|\;}

\def\Tall #1.{\forall #1.\spac}%
\def\Tmu #1.{\mu #1.\spac}%
\def\Tref(#1){\textlang{ref}(#1)}

\newcommand{\twoArgMacro}[2]{\newcommand{#1}[2]{#2 ##1 \spac ##2}}
\twoArgMacro{\isList}{\operatorname{isList}}
\twoArgMacro{\map}{\operatorname{map}}
\twoArgMacro{\listFilter}{\operatorname{listFilter}}
\twoArgMacro{\all}{\operatorname{all}}

\author{Daniel Anderson}
\affiliation{%
	\institution{Carnegie Mellon University}
	\city{Pittsburgh}
	\state{PA}
	\country{USA}
}
\email{dlanders@cs.cmu.edu}

\author{Guy E. Blelloch}
\affiliation{%
	\institution{Carnegie Mellon University and Google Research}
	\city{Pittsburgh}
	\state{PA}
	\country{USA}
}
\email{guyb@cs.cmu.edu}

\author{Zak Kent}
\affiliation{%
	\institution{Carnegie Mellon University}
	\city{Pittsburgh}
	\state{PA}
	\country{USA}
}
\email{zkent@cs.cmu.edu}

\author{Siddhartha Jayanti}
\affiliation{%
	\institution{Dartmouth College}
	\city{Hanover}
	\state{NH}
	\country{USA}
}
\email{siddhartha.visveswara.jayanti@dartmouth.edu}



\begin{document}
\title{Big Atomics:  Non-Blocking Algorithms with a Direct Fast Path}

    \begin{abstract}

Atomic operations, such as load, store, and compare-and-swap (CAS), on single words form the basis of multiprocessor programming. Atomics on multiple adjacent words, i.e., {\em Big Atomics}, are a powerful construct with wide utility in concurrent programming and algorithm design.  Unfortunately, there are no existing algorithms for big atomics that are simultaneously efficient in theory and practice.  Algorithms based on sequence locks work well in practice when processes are not interrupted, but are blocking, and hence not theoretically efficient, nor efficient in practice under contention and oversubscription. On the other hand, theoretically efficient non-blocking solutions (wait-free or lock-free) require a level of indirection and suffer a significant practical cost.

In this paper, we present two new algorithms for big atomics, which are non-blocking, but avoid indirection in most cases.   The first algorithm is wait-free and supports operations in time proportional to the size of the atomic, but requires storing two copies of each object. The second is lock-free and its memory usage is optimal within low-order terms.  Both use a practical caching scheme, in which updated values are temporarily stored using a level of indirection and then are replaced with the direct value.

We experimentally evaluate our algorithms, comparing them to 
traditional locks, sequence locks, hardware transactional memory (HTM), and indirection (wait-free). Comparisons are across a broad set of workloads. Our experiments show that our algorithms, and our lock-free memory-efficient algorithm in particular, perform competitively with the best blocking solution (sequence lock) in friendly conditions, but are significantly better when there are more software threads than hardware threads (oversubscription). To better understand how big atomics perform as part of an application, we also measure the performance of big atomics as part of \parlayhash{}, a state-of-the-art concurrent hash map.
\vspace{-.05in}

 \end{abstract}
    \maketitle

\section{Introduction}

Atomic operations on single-word variables form the basis of thread-safe concurrent programming, since they ensure sequential consistency for loads, stores, and read-modify-write operations, such as compare-and-swap (CAS). Furthermore, such operations are supported by all modern
multiprocessors.
%
%
Many concurrent algorithms, however, benefit from efficient {\em big atomics} that cover structures spanning multiple adjacent words.  
Applications include software transactional memory~\cite{HLMS03}, multiversioning~\cite{reed78,BG83,Postgres12,SQL13,neumann2015fast,Wu17,LKA17}, hash tables (in this paper), 
concurrent union-find~\cite{JTB19, JT21}, concurrent binary trees~\cite{Natarajan14}, snapshotting~\cite{BCW24}, and timestamping~\cite{BJW23} (see Section~\ref{sec:applications} for more information).

Despite these and other applications, big atomics have no direct hardware implementation, and no existing algorithms are both
efficient in theory and practice.
By theory we mean that they are robust in the face of adversarial schedulers (i.e., guarantee progress even when threads stall), and by practice we mean that they perform well under friendly schedulers (i.e., when threads run uninterrupted).   

\begin{table}[b]
\small

\begin{tabular}{|l|c|c|c|c|}
\hline 
\bf Approach & \bf Progress &\bf Space & \bf Indirect \\ \hline
Indirect (\indirect) & wait-free & $nk + L$ & always \\ \hline
Lock (\simplock)  & always block & $nk + N$ & never  \\ \hline
Sequence Lock (\seqlock) & block on race & $nk + N$ & never \\ \hline
{\em \cawf} & wait-free & $2nk + L$ & on prior race \\ \hline
{\em \came} & lock-free & $nk + L$ & on race\\ \hline
\end{tabular}
\vspace{.05in}
\caption{
Properties of various approaches to implement big atomics (Load and CAS).
The bottom two are from this paper.
In the space column $N = O(n)$ and $L = O(n + p^2k)$,
where $n$ is the number of big atomics, $k$ is the size of each, and $p$ is the number of threads. 
The wait-free algorithms both take $O(k)$ time.
The $O(p^2k)$ term accounts for $O(p^2)$ unreclaimed nodes of $O(k)$ words each.
\vspace{-0.2in}
}
\label{tab:properties}
\end{table}

\begin{figure*}[t]
    \centering

     \begin{subfigure}[t]{0.99\textwidth}
         \centering
         \includegraphics[width=0.99\textwidth]{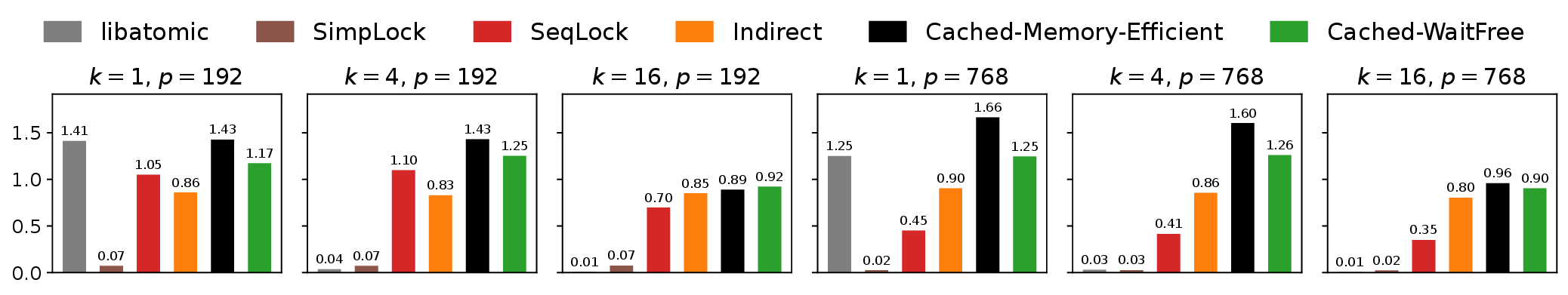}
     \end{subfigure}
    
    \caption{Geometric mean throughput in billions of operations per second from our big atomic implementations. The machine has 96 cores and 192 SMT threads. Geometric means are over all twelve combinations of 5\% and 50\% updates, uniform ($z = 0$) and high contention ($z = 0.99$), and $10^4$, $10^6$ and $10^8$ input size. The first three experiments use $192$ threads and the second three use $768$ threads. The experiments in each triple vary the size of the big atomic measured in number of $8$-byte words ($k$).}
    \label{fig:expsummary}
\end{figure*}

Existing \textbf{theoretically} efficient algorithms (wait-free or lock-free) that are robust under adversarial schedulers~\cite{Michael04,AndersonMoir99,JayantiP05,BlellochWei20} all require a level of indirection.  Such
indirection has large costs in practice due to the need for every operation to first read a pointer and then follow it through to an indirect
object.  Both accesses are likely to be cache misses, since the objects are shared, often doubling the cost. Furthermore, there are additional costs due to the need to use safe memory reclamation (SMR)~\cite{sheffi2023era,hazard04,fraser2004practical}.
Existing \textbf{practical} algorithms, on the other hand, are based on sequence locks~\cite{Hemminger02,Lameter05,Boehm12,Sullivan17}.  Sequence locks avoid indirection and hence outperform
 theoretically efficient algorithms by almost a factor of two under friendly conditions.  However, sequence locks
are blocking and perform badly when threads stall or are interrupted, or even under high contention.    We know of no prior non-blocking solution that avoids indirection.



In this paper, we develop two new algorithms supporting big atomics that are both theoretically efficient (non-blocking) and practical,
mostly avoiding indirection.
The first algorithm, \emph{\cawf}, supports wait-free load and CAS operations.   It requires two copies
of each object and is hence not particularly memory efficient.   The second algorithm, \emph{\came}, supports lock-free load and CAS operations.   Its memory requirement is the same as that for sequence locks as a function of the number of big atomics.  It is also
somewhat faster in practice as a result of its smaller memory footprint.
The properties of the algorithms along with those of various existing approaches for supporting big atomics are given in Table~\ref{tab:properties}.  Both our algorithms are built on atomic load-linked (\LL) and store-conditional (\SC) operations~\cite{Herlihy93,Moir1997Practical,jayanti2003efficient,BlellochWei20,originalllsc}, which
are useful operations on their own since they can avoid the ABA problems~\cite{treiber1986systems} that are common with load and CAS.
In our lock-free algorithm a load (\LL) will only use indirection if it races with a store (\SC), and in
the wait-free algorithm only if it either races with a store or if the prior
store raced with another store.

Our algorithms use a novel fast-path-slow-path approach that in the fast path can simply load the value directly.
The basic idea is to support both an indirect version and an inlined, a.k.a. direct {\em cached} version.  
When a process writes to a big atomic, it first creates an indirect version which it atomically links in via a pointer, and it then copies the indirect version into the inlined (i.e., cached) version. 
This allows readers to retrieve a consistent view of the value from the indirect version while it is being updated in the cached version, which could be blocked. 
Otherwise, it can use the fast path and read the cache directly.  The lock-free algorithm removes the indirect version before completing, while the wait-free one does not.   We leave it as an open question whether an algorithm can be both wait-free and memory efficient.

\myparagraph{Experiments}
To determine the efficiency of our algorithms compared to the prior approaches to big atomics, we implemented all the algorithms (in C++) and benchmarked them.
We compare several implementations of big atomics, including: \cawf, \came; our own implementations of big atomics based on locks, sequence locks, and indirection; as well as libstdc++ std::atomic (GNU libatomic, which uses locks), and a hardware transactional memory (HTM) based approach.

We run experiments over a range of workloads on a machine with 96 cores and 192 threads.  We first measure performance on operations applied to an array of big atomics.  The measurements are made under a varying mix of reads and CASes (ranging from all reads to all CASes), over a varying number of threads, over a varying number of big atomics (1K--100M),  over a range of contention levels (Zipfian parameter 0 to .99), and over a range of sizes for each big atomic (8--128 bytes).  

Figure~\ref{fig:expsummary} shows the geometric mean of throughput in several such settings in different implementations, different sizes of the big atomic,  and with 192 and 768 software threads.
With 768 threads the experiments are running with ``oversubscription'', meaning there are more software threads than hardware threads forcing the OS to interrupt and deschedule threads.
The experiments indicate that our algorithms outperform the others in most scenarios, and that the performance of the blocking \seqlock{} suffers significantly under oversubscription.   More detailed experiments varying each parameter are presented in Section~\ref{sec:benchmarks}.

%


\myparagraph{Hash Table}
To better understand how the variants of big atomics perform in an application, we use them in \parlayhash{}~\cite{parlayhash}, a growable concurrent hash table supporting arbitrary types.
The most efficient sequential hash tables such as Google's Swiss table~\cite{swisstable} and Facebook's F14~\cite{follyF14} gain their efficiency by using open-addressing to store key-value pairs directly in the array of slots, as opposed to requiring the indirection
used by chaining techniques.   Swiss table performs almost a factor of three better than C++ native hash tables, which are based on separate chaining~\cite{parlayhash}.
To implement such indirection avoidance in a concurrent setting, however, requires loading and updating
multi-word blocks atomically.
\parlayhash{} implements this by storing multiple fields within a single cache line and using big atomics with \LLSC{} to load and update the required information.

\begin{figure}
    \begin{subfigure}{\columnwidth}
      \centering
  		\includegraphics[width=0.95\columnwidth]{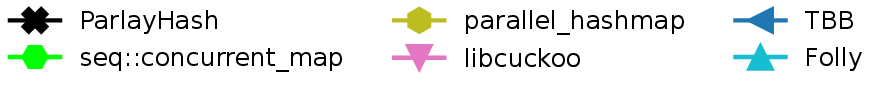}
        \bigskip
    \end{subfigure}

    \medskip

    \includegraphics[height=3cm]{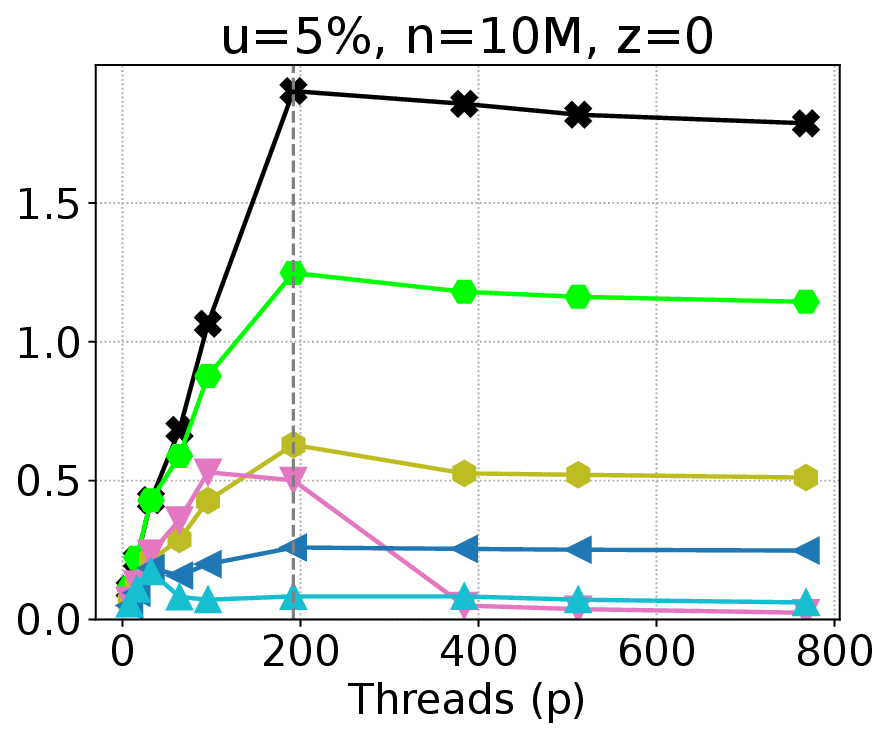}
    \includegraphics[height=3cm]{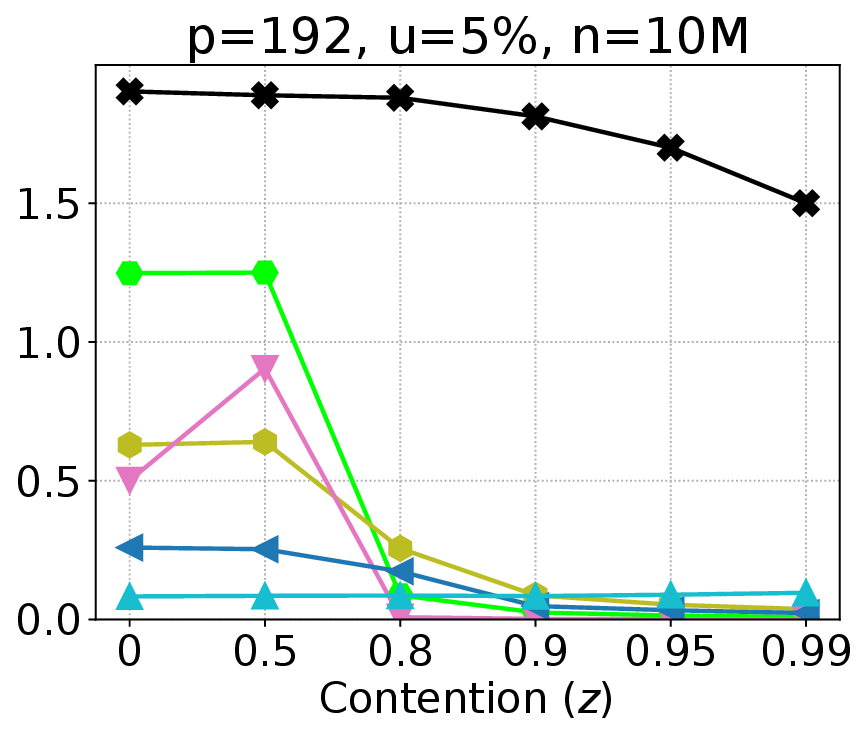}
    
    \caption{Throughput in billions of operations per second for our \ourHashtable{} hashtable (using our \came{} big atomic) against open-source hash tables, varying thread counts $(p)$ and contention parameter $(z)$.}
    \label{fig:hashtable-comparison}
\end{figure}

In our experiments, we compare the performance of \parlayhash{} when using the different algorithms for the big atomics.
We also compare to several widely used open-source library implementations~\cite{tbb07,onetbb,cuckoo14,libcuckoo,follyF14,BoostHash}, including those by Google~\cite{abseil}, Facebook~\cite{follyF14}, and Intel~\cite{tbb07,onetbb}.   This includes the fastest general concurrent hash table we are aware of, which is \emph{seq::concurrent\_map}~\cite{seqhash}.  Both sets of experiments are over a mix of parameters.
Figure~\ref{fig:hashtable-comparison} shows a cross section of the experiments comparing to the open-source library implementations.
The experiments indicate that \parlayhash{} using our \came{} big atomic outperforms the other hash tables and is robust both with oversubscription (left) and with high contention (right).   Comparisons across big atomic implementations are given in Section~\ref{sec:benchmarks}, Figure~\ref{fig:c7i-hash-list} and indicate \came{} outperforms the other big atomics in most cases.

\myparagraph{Our Contributions}
The contributions of our paper are:

\begin{itemize}[leftmargin=*, label=-]
    
    \item 
    The first non-blocking algorithms supporting big atomics that (mostly) avoid indirection.  
    {\em \cawf} is a wait-free algorithm, and
    {\em \came} is a more memory-efficient lock-free algorithm.
    They support load and CAS, or \LL{} and \SC{}.
    \item An {\em experimental comparison} of different implementations of big atomics across a broad set of workloads and algorithms.   
    \item We demonstrate their applicability by using them to implement an {\em efficient hash table}, whose performance we compare to other state-of-the-art hash tables.
\end{itemize}

Our code and benchmarks are available at\\ \url{https://github.com/cmuparlay/bigatomic}.
This paper is the full version of the conference paper~\cite{AndersonBKJ26}.   It includes Section~\ref{sec:applications}, the code and descriptions for VL (validate), more complete proofs, one small fix in the code for \cawf{}, and correcting the heap memory of \came{} to
be $O(p^2k)$ instead of $O(p^2)$.

    \section{Model and Preliminaries}
\label{sec:model}





We work in the standard asynchronous shared-memory multiprocessor model of computation.
We denote the set of processes by $P$ and their number by $p = |P|$.
In the paper we assume a machine that supports various atomic (sequentially consistent) operations on single words of memory.  
We assume a load (read), a store (write), and a CAS (compare and swap or compare and exchange).  

We use the standard definitions for \emph{wait-free}, \emph{lock-free} and \emph{linearizability}~\cite{HerlihyShavitBook}.   
Roughly speaking, wait-free means every thread will make progress as long as it is taking steps, and lock-free means at least one thread always makes progress even though some threads might be stopped.   A data structure is linearizable if operations act as if they happen atomically (at a single point in time) at some point between their invocation and response, despite concurrency.   
All of our algorithms are linearizable.  
Throughout the paper we use $n$ to indicate the number of big atomics, $p$ the number of threads, and $k$ the size in words of each big atomic.

We implement big atomics supporting Load/CAS.
A Load/CAS object $Y$'s state is simply a value $val$.
The object supports two operations:
\begin{enumerate}
    \item
    \underline{$Y.Load()$}:
    returns $Y.val$.
    \item
    \underline{$Y.CAS(exp, new)$}:
    if $Y.val = exp$, modifies $Y.val$ to $new$ and returns $\true$; otherwise, simply returns $\false$.
\end{enumerate}
As a stepping stone to the load/CAS big atomics, we implement \emph{LL/SC big atomics} supporting the load-linked (LL), validate (VL), and store-conditional (SC) primitives.
We give a formal sequential specification of the LL/SC objects and recall the famous wait-free reduction from load/CAS to LL/SC below.

The \emph{(abstract) state} of a collection of $n$ LL/SC objects $X_1,\ldots,X_n$ for processes in the set $P$ is defined as follows.
Each object $X_i$'s state is a pair $(val, pset)$, where $val$ represents the value of the object and $pset$ represents the set of processes that have the object \emph{linked}.
The operations are defined below.
\begin{enumerate}
    \item
    \underline{$X_i.LL_\pi()$}:
    removes $\pi$ from $X_j.pset$ for each $j \in [n]$, adds $\pi$ to the set $X_i.pset$, and returns $X_i.val$.
    \item
    \underline{$X_i.VL_\pi()$}:
    returns whether $\pi \in X_i.pset$.
    \item
    \underline{$X_i.SC_\pi(new)$}:
    if $\pi \in X_i.pset$, modifies the state to $(new, \varnothing)$ and returns $\true$; otherwise, simply returns $\false$.
\end{enumerate}

A collection of $n$ CAS objects $Y_1,\ldots,Y_n$ can be implemented from a collection of $n$ LL/SC objects---initially $\forall i: X_i.val = Y_i.val$ and $X_i.pset = \varnothing$---by the following simple wait-free reduction.

\begin{algorithm}[h]
\vspace{-.1in}
\begin{lstlisting}[basicstyle=\footnotesize\ttfamily,linewidth=.99\columnwidth, xleftmargin=5.0ex,numbers=left]
T Load(int i) { return X[i].LL(); } 

bool CAS(int i, T exp, T nu) {    
  T old = X[i].LL();    
  if (exp != old) return false;
  else if (exp == nu) return true;
  else return X[i].SC(nu);}
\end{lstlisting}
\vspace{-.1in}
\caption{Reduction from load/CAS to LL/SC. Code for $Y_i$.}
\label{alg:load-cas-from-ll-sc}
\end{algorithm}
\vspace{-.1in}


    


    

In the remainder of the paper, we design wait-free and lock-free LL/SC big atomics.

\section{Background and Related Work}
\label{sec:prior}

There are two classes of approaches to implement big atomics supporting \LLSC{} or Load/CAS: \textbf{lock-based} and \textbf{indirect}.   
The lock-based methods are blocking, while the indirect methods require that all loads read through a pointer, and updates allocate space for new objects.  We know of no prior approach that avoids blocking while also avoiding indirection.  

The simplest lock-based method for big atomics, which we will refer to as \simplock{}, is to associate a lock with each big atomic, and wrap the sequential code with the lock.   In particular, a load reads the value under a lock, and a CAS reads the value, checks equality, and conditionally writes, all under the lock.   To implement \LLSC{} one can add a sequence number to each big atomic, and increment it every time the object is updated, also under the lock.
The performance of \simplock{} can be poor, as even loads compete for the lock among themselves.  
Reader-writer locks could, in principle, be used to reduce contention among loads, but they are expensive~\cite{bravo19,abseil_tip_197}.

A more sophisticated lock-based method is \emph{sequence locks}~\cite{Hemminger02,Lameter05,Boehm12,Sullivan17} (\seqlock{}).  
A sequence lock associates a sequence number with each big atomic.  
An update locks the atomic by incrementing the number from an unlocked state (even) to a locked state (odd), making the change, and then unlocking by incrementing again to even.
Loads read the sequence number, then the value, then the number again.   If the numbers changed or the number is odd (locked) then it has to try again, otherwise it returns the value.  Implementing \LLSC{} is straightforward, since the sequence number seen during an \LL{} can be checked for equality at the matching \SC{}.
With sequence locks, loads
can be blocked indefinitely by a stalled updating thread.
In practice, however, sequence locks can be very much more efficient than simple locks especially under load-heavy contention since loads do not block each other.


The second general approach to implement big atomics is to use a level of indirection (\indirect{}).  In this approach, each atomic variable holds a pointer to a heap-allocated object holding the value.  A load reads the pointer and 
returns the value from the object, and, if it is an \LL{}, saves the pointer to check later.
An \SC{} allocates a new object on the heap with the new value, and then uses a CAS to try to update from the
pointer left by the \LL{} to the new object pointer.  If successful, it can delete the old object returning true, and
if not, it can delete the new object returning false.
Safe memory reclamation, such as hazard pointers~\cite{hazard04} or epoch-based collection~\cite{fraser2004practical} is required to
implement this correctly so that it avoids the ABA problem~\cite{treiber1986systems} and read-destruct races~\cite{anderson2022turning,hazard04}.
Hazard-pointers, for example, provide an interface where a pointer can be "protected" from reclamation, and deletes are
replaced by "retires", which queue an object for reclamation but ensure that the object is not reclaimed and reused until it is no longer protected.

These traditional SMR approaches, however, are not sufficient if one wants both
bounded memory and wait-free operations in constant time. Epoch-based collection does not bound memory since
a stalled process can indefinitely delay reusing all memory.   Although 
hazard pointers bound the extra needed memory to $O(p^2k)$ for $p$ processors and objects of size $k$, they
are not wait-free
since they involve copying a pointer to an announcement slot, then checking if the pointer has changed, and
repeating if so.  
There are several algorithms that are wait-free and achieve various bounds on memory usage, under various assumptions~\cite{Michael04,AndersonMoir99,JayantiP05,BlellochWei20}.   The latest supports constant time operations, using only single-word CAS, with the space given in Table~\ref{tab:properties}, and no need for unbounded counters~\cite{BlellochWei20}.  
The approach is based on fixing the problem with hazard pointers using a constant-time operation
that atomically copies the pointer from a location into the announcement slot, avoiding retrying.   It allows for a read after
delete conflict.   This is safe in theory, but might not be if a delete
returns reclaimed memory to the operating system.
We use hazard-pointers in our algorithms, and for our wait-free algorithm, we assume it uses the 
constant-time atomic-copy operation.

%

From a performance perspective, the problem with all forms of indirection is that every load needs to first read a location, and then follow the stored pointer to a second allocated location.  Both of these could be a cache miss, and furthermore the second depends on the first so the two reads cannot be pipelined.  In practice, the extra level of indirection can slow down performance by up to a factor of two.   A further inefficiency is that the SMR schemes copy a value (pointer or epoch number) to an announcement slot and then require a memory fence to ensure the
value is visible to all other threads before proceeding.  This can be as expensive as a cache miss~\cite{fences2011}.
Loads with sequence locks require neither the indirect load nor the memory fence.   In our algorithms we avoid both in the fast path, but incur them both in the slow path.  The idea of avoiding fences for memory protection in a fast path has been used elsewhere~\cite{Singh2025Publish,bragagnolo2016fast}.

Big atomics can be implemented with
hardware Transactional Memory (HTM)~\cite{htm}. In principle, they are the ideal application of HTM since often a big atomic will reside on a single cache line requiring the HTM mechanism to only track one line.  HTM, however, always requires a software backup mechanism if the transaction fails, which requires either locks or indirection.      
In Section~\ref{sec:benchmarks} we run experiments with HTM. We find that it never outperforms the software-only approaches, even on low contention where HTM transactions almost always succeed. 
Furthermore, HTM has been disabled on new processors since 2021 due to security issues~\cite{htmdead}.  

A multiword-CAS (MCAS)~\cite{harris2002practical,guerraoui2020efficient} is a CAS that operates
 on multiple locations, the swap occurring in all locations if all locations match.   It is more powerful than a big atomic, since the locations need not be adjacent.  A big atomic CAS can be implemented with an MCAS across the adjacent words; however, MCAS requires at least $k$ single-word CAS operations~\cite{guerraoui2020efficient} and is therefore unlikely to be practical as an implementation for a big atomic CAS.

    \section{Algorithms}

We now describe our two algorithms for big atomics which share the following characteristics.   They maintain both a ``\emph{cached}'' (inline) copy of the data, and a pointer to a heap-allocated ``\emph{backup}'' copy.  Loads first try to read from the cached copy (fast path) and only go to the backup (slow path) if necessary.    Updates first install the backup copy in a state that indicates the cache is invalid, and then attempt to copy the data into the cached copy.  

In our first algorithm (\cawf) 
\LL{} and \SC{} take $O(k)$ time.   However, this algorithm requires that the backup copy is always populated, causing a factor of two in wasted space.   Our second algorithm (\came) avoids the extra space by retiring the backup after the cache is installed.  
It also only uses one header word instead of two.
However, this comes at the cost of making the \LL{} and \SC{} lock-free instead of $O(k)$ time.  Algorithmically, a notable difference in the second algorithm is that updates need to help other updates complete. 

We present C++ code for our algorithms.  In the code, we use the following API for hazard pointers, which
is similar to the API added to C++26~\cite{hpcpp26}.    A \lstinline{hazard_slot} supports a \lstinline{protect},
to protect 
a pointer stored in a given location, and a \lstinline{protect_direct} to protect a pointer that
is not in a location. We assume that a \lstinline{protect} or \lstinline{protect_direct} unprotects the prior pointer in the slot.
A \lstinline{retire} is used to safely delete objects that might be protected by another thread.
For our wait-free algorithm, we assume that \lstinline{protect} uses a
wait-free atomic copy to ensure that it is wait-free (discussed in Section~\ref{sec:prior}).   We assume that all memory operations are sequentially consistent.  In our actual code, we include additional fences
and relaxed versions of loads and stores where it is safe.   To avoid line overflow, we use \lstinline{cmp_xch} as a shorthand for C++ \lstinline{compare_exchange_strong}.   


\begin{algorithm*}[t]
\vspace{-.2in}
\begin{multicols}{2}
\begin{lstlisting}[basicstyle=\footnotesize\ttfamily,linewidth=.99\columnwidth, xleftmargin=5.0ex,numbers=left]
template<T>                                                     
struct BigAtomic {                                              
  struct backup {          
    T value; };                                                 
                                            
  using tag = uint64_t; // used to pass info from LL to SC        
  std::atomic<uint64_t> seqnum; @\label{fig:wfrc:version}@
  std::atomic<backup*> backup_ptr;                                
  T cache; @\label{fig:wfrc:cachev}@
  static thread_local inline tag expected_tag = 0;              
  static thread_local inline hazard_slot h1;
  static thread_local inline hazard_slot h2;
                                                                
  bool is_seqnum(tag); // tag can be a seq num or a pointer                                      
  tag to_tag(backup*);                                          
  backup* to_ptr(tag);                                          
  bool is_marked(backup*);   @\label{fig:wfrc:ismarked}@                                    
  backup* mark_ptr(backup*); // mark pointer as invalid             
  backup* unmark_ptr(backup*); // strip the mark @\label{fig:wfrc:unmark}@
public:                                                                                                      
  BigAtomic() : seqnum(0),                                                                                   
                backup_ptr(new backup(T{})), cache{} {}                                                               
  T LL() {                                                                                                   
    auto seq = seqnum.load();                                                                                
    T val = cache;  // bytewise atomic load
    backup* p = backup_ptr.load(); @\label{fig:wfrc:loadbackup}@
    if (!is_marked(p) && seq == seqnum.load()) { // fast path @\label{fig:wfrc:testvalid}@ 
      expected_tag = seq;  // save sequence number to check in SC
      return val; } // return value from cache
    p = h1.protect(backup_ptr); // load and protect  @\label{fig:wfrc:loadbackupprotected}@                      
    expected_tag = to_tag(p); // save pointer to check in SC   
    return unmark_ptr(p)->value; }  // get value from backup

  bool VL() {                                                                                                             
    backup* p = backup_ptr.load(); @\label{wf:vl:backupload}@
    if (is_seqnum(expected_tag)) {  @\label{wf:vl:isseq}@
      auto seq = seqnum.load(); @\label{wf:vl:seqload}@
      return (!is_marked(p) && seq == expected_tag); @\label{wf:vl:seqreturn}@
    } else return // expected_tag is a pointer
        unmark_ptr(p) == unmark_ptr(to_ptr(expected_tag)); } @\label{wf:vl:ptrreturn}@  

  bool SC(const T& desired) {                                                                                
    int64_t seq;                                                                                             
    backup* p;
    // Set seq and p depending on path taken by LL
    if (is_seqnum(expected_tag)) {   // fastpath                                                                        
      p = h1.protect(backup_ptr); @\label{fig:wfrc:backupread}@   
      seq = seqnum.load();                                                                                   
      if (is_marked(p) || seq != expected_tag) return false; @\label{fig:wfrc:notexpected}@
    } else { // slowpath: expected_tag is a pointer
      p = to_ptr(expected_tag);                                                                              
      seq = seqnum.load(); }
    auto ptr = new backup{desired}; // allocate  @\label{fig:wfrc:allocstart}@
    h2.protect_direct(ptr);
    auto new_p = mark_ptr(ptr); @\label{fig:wfrc:allocend}@
    auto old_p = p;   
    p = backup_ptr.load();
    if ((p == oldp && @\label{fig:wfrc:readagain}@
         backup_ptr.cmp_xch(p, new_p)) // try to install backup   @\label{fig:wfrc:install}@
        || (p == unmark_ptr(old_p) &&                                                                        
            backup_ptr.cmp_xch(p, new_p))) { @\label{fig:wfrc:retry}@ // try again
      retire(unmark_ptr(p));
      // try to install cached value
      if ((seq % 2 == 0) && seq == seqnum.load() &&   @\label{fig:wfrc:testversion}@   
          seqnum.cmp_xch(seq, seq + 1)) { @\label{fig:wfrc:lock}@
        cache = desired;   // bytewise atomic store
        seqnum.store(seq + 2); @\label{fig:wfrc:unlock}@
        backup_ptr.cmp_xch(new_p, unmark_ptr(new_p)); } @\label{fig:wfrc:clear}@                                                     
      return true; }  @\label{fig:wfrc:failcache}@                                                                                     
    delete unmark_ptr(new_p); // installation of backup failed
    return false; }
};
\end{lstlisting}
\end{multicols}
\vspace{-.2in}
\caption{A wait-free mostly-direct big atomic supporting \LLSC{}.
\vspace{-.1in}}
\label{alg:wait-free-load-cas}
\end{algorithm*}

\subsection{Cached Wait-Free}\label{sec:wait-free-load-cas}

Code for our first algorithm, \cawf{}, is given in Algorithm~\ref{alg:wait-free-load-cas}.   For each big atomic, the algorithm maintains a sequence number, a pointer to the indirect version, which we will refer to as the \emph{backup} node, and the inlined version, which we will refer to as the \emph{cache} (Lines~\ref{fig:wfrc:version}--\ref{fig:wfrc:cachev}).
 To help maintain consistency between the backup and the cache, the algorithm maintains an extra mark on the backup pointers (e.g., stolen from the least-significant-bit), which if set indicates the cache is invalid (Lines~\ref{fig:wfrc:ismarked}--\ref{fig:wfrc:unmark}).  The basic idea of the algorithm is on an update to first install a backup with the mark on the pointer indicating that the cache is invalid, then to use a sequence lock to update the cache, and finally to clear the mark on the backup
 pointer.   The fast path of a load is similar to that of sequence locks.
 To pass information from an \LL{} to an \SC{} we use the \lstinline{expected_tag} variable.  It can store either a sequence number or a pointer, and the function \lstinline{is_seqnum} is used to determine which.
 
 The algorithm has the following properties: (1) the logical state of a big atomic is the value stored in the backup node, (2) the
 sequence number never decreases, (3) when the backup pointer is valid (not marked), the cache and backup contain equal values, and the sequence number is even (unlocked), (4) the cache is never modified when the sequence number is even, and (5) a pointer going from unmarked (valid) to marked (invalid) implies a logical update. 
 
\myparagraph{The \LL{} Operation}
An \LL{} starts by running steps similar to those required by sequence locks, i.e., reading the sequence number, then the cache, and then the sequence number again.   However, in addition to reading the cache before the second read of the sequence number, it also reads the backup pointer (Line~\ref{fig:wfrc:loadbackup}). Now, it checks that the backup pointer is not marked and that the sequence number has not changed (Line~\ref{fig:wfrc:testvalid}).   If both are true, the cache was read during a period of time in which it was not being updated and for which the cache was valid, and we can therefore return the value.
The algorithm also stores the sequence number in \texttt{expected\_tag} to be checked at the next \SC{}.
This is the fast path---no indirect values were accessed, and no memory needed to be protected.  
Note that standard sequence locks would check whether the sequence number is odd to ensure that it is unlocked.   We need to check not just that it is unlocked, but also that the cache is valid.   By properties 2 and 3 and given that the sequence number has not changed, an unmarked backup pointer tells us both.

If either the backup pointer is invalid or the sequence number has changed, then the algorithm has to go to the slow path and read the backup node.   Before reading through to the backup node, the algorithm protects it with a hazard pointer (Line~\ref{fig:wfrc:loadbackupprotected}).  This is necessary to avoid load-reclamation races and ABA issues. It ensures
that the node is not reclaimed until after the next time \texttt{h1} is used for protection.   This slow path stores the protected
backup pointer in \texttt{expected\_tag} to be checked at the next \SC{}.

\myparagraph{The \VL{} Operation}
A \VL{} first starts by loading the current backup into \texttt{p}
on line~\ref{wf:vl:backupload}, and then cases on whether \texttt{expected\_tag}
is a seqnum or a pointer. If a seqnum, it loads the value of the current
sequence number and returns \texttt{true} iff the value of the sequence number
has not changed and \texttt{p} is unmarked. If \texttt{p} is a pointer,
it unmarks both \texttt{p} and \texttt{expected\_tag} and returns \texttt{true}
iff they are equal.

\myparagraph{The \SC{} Operation}
An \SC{} uses the \texttt{expected\_tag} to determine if the logical state has been updated and to update if not.
It first sets up an expected pointer \lstinline{p} and an expected sequence number \lstinline{seq}.  If
the expected tag is a sequence number, then the associated LL took the fast path and never protected the backup pointer.
The algorithm therefore does a protected load on the backup pointer.   It then reads the current sequence number and
checks that the pointer is valid and that the sequence number has not changed.   Note that due to property 5,
the pointer being invalid implies that the logical state has changed.  If the expected tag is a pointer, then the associated \LL{} took the slow path, and the pointer is already protected.   The algorithm reads only the current sequence number.
At this point, in both cases, the algorithm has a protected backup pointer and a sequence number that was read after the pointer. 

Now, the algorithm tries to install the new desired value, by first trying to install a new backup node with the value 
and if successful trying to copy the value to the cache.
The new node is marked as invalid since when first installed as a backup, the cache is not valid.   If the backup is successfully installed, and the value is successfully copied into the cache, then we attempt to remove the invalid mark on the backup pointer (Line~\ref{fig:wfrc:clear}).
The new backup node is allocated, protected, and its pointer is
marked on Lines~\ref{fig:wfrc:allocstart}--\ref{fig:wfrc:allocend}.

The algorithm then uses a compare-and-exchange (\lstinline{cmp_xch}) to attempt to install the backup using \lstinline{p} as the expected
value.  The \lstinline{cmp_xch} might need to be attempted twice since the first attempt could fail due to the old backup pointer being converted from invalid to valid between when the old pointer was read and the \lstinline{cmp_xch}.  Hence, if the first try (Line~\ref{fig:wfrc:install}) fails on an invalid pointer, we try again, but with the validated version of the old pointer (Line~\ref{fig:wfrc:retry}).  If both attempts fail, the backup installation has failed due to another thread modifying the logical state.
 Therefore, it deletes the allocated backup node and returns false.   If a \lstinline{cmp_xch} attempt succeeds, the installation has succeeded so 
the algorithm retires the old backup node, and attempts to install the cached value.  

To install the cached value, the algorithm uses an approach similar to sequence locks: it tries to increment the sequence number from even to odd (Line~\ref{fig:wfrc:lock}), which takes a lock on the cache, and if successful, copy the desired value into the cache, and unlock it with a second increment (Line~\ref{fig:wfrc:unlock}).  
A subtle difference from sequence locks is that, in addition to checking that the sequence number is even, it also checks that the sequence number has not changed since before the backup node was installed. This ensures
that the value of a more recently installed backup node is not overwritten.

When there is a concurrent update, the validation might fail due to the failure to acquire
the lock on Line~\ref{fig:wfrc:lock} or to validate the backup pointer on Line~\ref{fig:wfrc:clear}.  This can leave the cache invalid after the \lstinline{cmp_xch} completes.
However, validation will always succeed if not concurrent with another \lstinline{cmp_xch}.

We note that the code shown only supports a single outstanding \LL{} since the expected tag and the two hazard slots h1 and h2 are shared per thread (static thread-local members in C++).  It could be extended to support more than one outstanding \LL{}s by having more than one copy of each, e.g., using an array for each.

\vspace{-.05in}
\begin{theorem}
\label{theorem:wf}
The big atomic object described in Algorithm~\ref{alg:wait-free-load-cas} has linearizable \LL{}s and \SC{}s.
Furthermore for $n$ big atomics each of size $k$, and $p$ processes, all operations take $O(k)$ time, and the total memory usage is $2nk + O(n + p^2k)$.
\end{theorem}
\vspace{-.1in}

\begin{proof}[Proof Sketch]
For correctness,
we first identify the linearization points.  For the \LL{}, if it succeeds on reading the cached value the linearization point is at the load of the backup pointer on Line~\ref{fig:wfrc:loadbackup}, and otherwise
it is at the protected load of the backup pointer on Line~\ref{fig:wfrc:loadbackupprotected}.
The \VL{} linearizes unconditionally at the load of the backup on line~\ref{wf:vl:backupload}.
For the \SC{}, if it succeeds then
it linearizes at the \lstinline{cmp_xch} that sets the backup, and otherwise it linearizes
at the load or \lstinline{cmp_xch} that causes it to fail (one of lines~\ref{fig:wfrc:backupread}, \ref{fig:wfrc:readagain}, \ref{fig:wfrc:install}, or \ref{fig:wfrc:retry}).  The linearization points for a successful \SC{} imply property 1.


We now argue property 3---i.e. if the pointer is valid, then the cached copy equals the backup copy and the sequence number is even.
An \SC{} only validates its backup pointer after a copy of its value has been installed in the cache, and the validation only succeeds if the backup pointer has not changed since the \SC{} installed it.  Note this relies on the pointer being protected so that it is not recycled avoiding an ABA problem.
No earlier \SC{} could lock or modify the cache between when it is unlocked and
the validation of the backup pointer since this earlier \SC{} would fail the \lstinline{seq == seqnum.load()} test on Line~\ref{fig:wfrc:testversion}.   Hence when the pointer is validated, the cache is indeed valid with the same value in the cache as in the backup pointer, and the sequence number is unlocked (even).    Furthermore the cache has been valid since the \SC{} unlocked the cache by incrementing the sequence number.

If an \LL{} succeeds reading from the cache, it must have seen a valid backup pointer.   By property 3 the cache is not locked at that
point, and hence if the two loads of the sequence number (on either side) are the same, the cache was not locked and did not change.
Also by property 3, the cache equals the backup during that time so the load is correct.
If the reading of the cache fails, reading the backup pointer will get a correct value since updates linearize on a successful installation of the backup pointer (i.e. the value in the backup pointer is always the current value).
If an \LL{} fails reading from the cache, then it reads from backup, with protection, which contains the correct value.

We next consider the correctness of \VL{}. First suppose that the check on line~\ref{wf:vl:isseq} succeeds, and hence \texttt{expected\_tag} is a sequence number. In this case, the prior \LL{}  linearized at the load of an unmarked backup. We now prove that no \SC{} linearized between the LP of the prior \LL{} and the load of the backup into \texttt{p} on line~\ref{wf:vl:backupload} iff \texttt{p} is unmarked and the sequence number is still equal to \texttt{expected\_tag}. First suppose some \SC{} has linearized by installing a new backup pointer, which \texttt{p} witnesses. Then, because this new backup is unmarked only after validating the cache, the sequence number must have changed such that the sequence number observed on line~\ref{wf:vl:seqload} is distinct from \texttt{expected\_tag}. Now suppose, for the remaining direction, that either \texttt{p} is marked or the sequence number has changed. If \texttt{p} is marked, then by property (5) some \SC{} must have linearized since that of the prior \LL{}. If the sequence number has changed, then some newly installed backup must have been validated, also a logical update.

Now suppose that instead \texttt{expected\_tag} is a backup pointer,
which must have been protected. A \SC{} linearized by installing a new backup, which cannot be equal to \texttt{expected\_tag} since it was protected. Hence a \SC{} linearized iff the unmarking of \texttt{p} is not equal to the unmarking of \texttt{expected\_tag}. Thus line~\ref{wf:vl:ptrreturn} returns \texttt{true} exactly when no \SC{} has linearized since the prior \LL{}.

The \SC{} fails if and only if the logical state has changed.   By property 5, failing on Line~\ref{fig:wfrc:notexpected} implies the logical state has changed since at the \LL{} the pointer was valid, but now it is either invalid or the sequence number changed, which implies it was invalid at some point between the \LL{} and now.
If both compare and exchanges fail when installing the backup pointer, then the value in the backup pointer must have changed to point to a new backup, and hence the logical state changed, and the \SC{} correctly fails.
If the \SC{} successfully installs its backup node, it linearizes at that
point.   It will now attempt to lock the cache and only succeed if the cache is not already locked, and if the version has not changed since before the installation of its node.   This means no later successful \SC{} has updated the cache.  If the lock is successfully acquired, the backup pointer cannot be validated until the lock is released since no other CAS could succeed to update the cache, which is required to validate the backup pointer.

For time we note that the only loop is over the word length, taking $O(k)$ time to copy it in and out.
  The only non-trivial functions the algorithm calls are, protecting with the hazard pointer, and the memory management operations (new, retire, and delete).  Safe memory reclamation with hazard pointers can be made constant time using an atomic copy~\cite{BlellochWei20}, and the memory management can be made constant time assuming an upper bound on the number of big atomics is known ahead of time~\cite{blelloch2020concurrent}.  

For space we note that the space is taken by the objects themselves each of which is of size $k + 2$ (the data, the pointer and the counter), and the nodes.
There are exactly $n$ nodes installed at all times, each of size $k + 1$.
Moreover, at most $p$ BigAtomics may be announced simulateously, and a thread scans its retire list at least every $p$ iterations.
Hence the number of live backups is at most $O(p^2)$, and thus the total space incurred by memory management is $O(p^2 k)$.
The total is therefore bounded by $2nk + O(n + p^2 k)$.
\end{proof}




    \subsection{Cached Memory Efficient}\label{sec:lock-free-load-cas}

Our second algorithm optimizes the memory used by the big atomic. In particular, it reduces the number of header words from two to one, and it avoids the need to keep a backup on every big atomic.  For $p$ threads, the algorithm never has more than $p$ backup nodes installed at any given time, and the pool of backup nodes can be bounded by $O(p(1 + p/k))$ using hazard pointers.
This comes at the cost of making it lock-free instead of $O(k)$ time, however, its performance in practice almost always surpasses that of the wait-free algorithm due to memory savings (the pool of backup nodes typically fits in the last level of cache).   The code is given in Algorithm~\ref{alg:lock-free-single-header}. 

 \begin{algorithm*}[t]
\vspace{-.2in}
\begin{multicols}{2}
\begin{lstlisting}[basicstyle=\footnotesize\ttfamily,linewidth=.99\columnwidth, xleftmargin=5.0ex,numbers=left]
template <typename T>
class BigAtomic {
  struct backup {
    T value; long seqnum; };    

  using tag = uint64_t; // either a seqnum or pointer to backup
  std::atomic<tag> header; // single word header
  T cache; // the cached value
  static thread_local inline hazard_slot h1;
  static thread_local inline hazard_slot h2;
  static thread_local inline tag expected_tag = 0;

  bool is_pointer(tag);
  tag& to_tag(backup*&);
  backup* to_pointer(tag);
public:
  BigAtomic() : header(0), cache{} { }
  
  T LL() {                                    
    while (true) {                                            
      tag hdr1 = header.load();  @\label{lf:ll:hdr1}@                          
      T val = cache; // bytewise atomic load @\label{lf:ll:cache}@                  
      tag hdr2 = header.load(); @\label{lf:ll:hdr2}@       
      if (!is_pointer(hdr1) && (hdr1 == hdr2)) {  @\label{lf:ll:fastpathcheck}@
        // fast path
        expected_tag = hdr1; // save seq number to check in SC
        return val; } // return value from cache
      // slow path
      tag hdr = h1.protect(header);  // load and protect @\label{lf:ll:protect}@
      if (is_pointer(hdr)) { // successfully read a pointer @\label{lf:ll:ispointer}@
        expected_tag = hdr; // save pointer to check in SC
        return to_pointer(hdr)->value; }}} // value from backup

  bool VL() {
    tag actual_tag = header.load(); @\label{lf:validate:load}@
    if (expected_tag == actual_tag) { @\label{lf:validate:eqcheck}@
      return true; } // equal seqnums or pointer
    if (is_pointer(expected_tag) && 
        to_pointer(expected_tag)->seqnum == actual_tag)  @\label{lf:validate:validatecheck}@
      return true;  // backup pointer was validated
    return false; } @\label{lf:validate:fail}@

  bool SC(const T& v) {                     
    // STAGE 1: allocate and protect backup 
    long seqnum = is_pointer(expected_tag) ? // set seqnum 
        to_pointer(expected_tag)->seqnum :
        expected_tag;
    long new_seqnum = seqnum + 1;
    backup* new_backup = new backup{v, new_seqnum}; @\label{lf:sc:backup_alloc}@  @\label{lf:new}@            
    h2.protect_direct(new_backup);  @\label{lf:sc:backup_protect}@                                       
    // STAGE 2: try to install backup
    tag cur_tag = header.load()@\label{lf:sc:load-cur-tag}@ ;   @\label{lf:start2}@
    if (cur_tag != expected_tag ||                                 
        !header.cmp_xch(cur_tag, to_tag(new_backup))) {     @\label{lf:cas1}@     
      if (is_pointer(cur_tag) || (cur_tag != seqnum) ||       
          !header.cmp_xch(cur_tag, to_tag(new_backup))) {  @\label{lf:cas2}@
        allocator::destroy(new_backup); // installation failed   
        return false; }  } @\label{lf:end2}@
    if (is_pointer(cur_tag)) //was pointer, do not install cache    
      retire(to_pointer(expected_tag));                              
    else { // STAGE 3: try to install the cached value @\label{lf:start3}@
      cache = v; // bytewise atomic store @\label{lf:copywrite}@
      while (header.load() != to_tag(new_backup) ||                
             !header.cmp_xch(to_tag(new_backup), new_seqnum)) { @\label{lf:cas3}@
        new_backup = to_pointer(h1.protect(header));
        cache = new_backup->value; // bytewise atomic store
        new_seqnum = new_backup->seqnum;  }                                                       
      retire(new_backup); }  @\label{lf:end3}@
    return true; }      
};              
\end{lstlisting}
\end{multicols}
\vspace{-.2in}
\caption{A lock-free big atomic supporting \LL, and \SC.  The memory requirement is $O(p^2k)$ additional heap memory for $p$ threads, independent of the number of big atomics.
\vspace{-.1in}}
\label{alg:lock-free-single-header}
\end{algorithm*}



The algorithm uses a single word as a header that can                          
be either a sequence number or a pointer, using the high-order                             
bit to distinguish them.  It maintains the cache in a similar way as                             
with the wait-free algorithm, but differs from it in two important                          
ways.  Firstly, since only one of either a sequence number or pointer                           
is installed at any time, a load might see two sequence numbers                          
that have changed, but not see a pointer from which to extract a value,  and hence
have to repeat.
However, each retry can only happen if an update occurred and is therefore                        
lock free.  Secondly, to maintain space efficiency, the algorithm                            
needs to ensure that the backup pointer is removed.  This requires the 
\SC{} to keep trying if another thread updates in the meantime. Again, this is                             
lock free since it only needs to repeat if another update has succeeded.                            
                                                                                                
We define \emph{\seqnumber} as the header if it is a sequence number, otherwise it                            
is the sequence number stored indirectly through the pointer in the backup node.  We                               
define \emph{val} to be the value in the cache                            
if the header contains a sequence number, and otherwise the value                               
stored indirectly through the pointer.                                                          
The algorithm has the following properties, which are justified in the proof of Theorem~\ref{proof:lf}:

\begin{enumerate}[leftmargin=*]
    \item\label{prop:logstate} The logical        
state of a big atomic is its val, as just defined.
    \item\label{prop:seqincr} Every logical
update (successful \SC) of a big atomic increments its \seqnumber{} by one.
    \item\label{prop:noupdate} The cache is never modified when the header contains a sequence number.
    \item\label{prop:sequpdate} If the header is a \seqnumber{}, it can only change to a pointer.
    \item\label{prop:ptrupdate} If the header is a pointer, it can change to another pointer via a logical update, or to a sequence number with the same val and \seqnumber{} (this case does not involve a logical update).
    \item\label{prop:ptrtoseq} The header is only ever changed from a pointer to a sequence number by an \SC{} that had previously changed the header from a sequence number to a pointer.
    \item\label{prop:ll} The thread-local \texttt{expected\_tag} field is value of the header at the linearization point of the previous \LL{} performed by that thread. Further, if \texttt{expected\_tag} is a pointer, it is protected.
    \item\label{prop:freshptr} Only newly allocated pointers are used for backups, and they are readonly after initialization.
    \item\label{prop:exclown} Only the thread that last installed a backup pointer in place of a sequence number may modify the cache.
\end{enumerate}

\myparagraph{The \LL{} Operation}
An \LL{} has a fast path similar to that of the wait-free                                  
algorithm.  It reads the header, reads the cache, and then                           
reads the header again.  If both headers are equal and are a sequence
number  it sets the \texttt{expected\_tag} to the sequence number and returns the value from
the cache.  This is correct by properties                          
(1), (2), and (3) above.  If both reads of the header are sequence numbers                            
but differed, it has to try again.  If either was a pointer, then it                             
needs to attempt the slow-path. Since the fast path loads were
not protected, the algorithm reloads the header, this time with hazard pointer 
protection.  If this load returns a sequence number, it has to go back                          
to the start, but if it returns a pointer to a backup, now protected,
it sets \texttt{expected\_tag} to the pointer and returns its value.
The hazard pointer is maintained until the next \SC{} or \LL{} by
the same thread.

\myparagraph{The \VL{} Operation}
The \VL{} operation is fairly simple. It first loads
the value of the current tag into \texttt{actual\_tag}
on line~\ref{lf:validate:load}, and compares it to
\texttt{expected\_tag} on line~\ref{lf:validate:eqcheck}. If it succeeds
then both are the same pointer or same sequence number,
in which case no intervening \SC{} could have occurred
since the prior \LL{} and \texttt{true} is returned.
Otherwise, if \texttt{expected\_tag} is a pointer,
it checks if has been validated on line~\ref{lf:validate:validatecheck}.
If so, then again no intervening update has occurred, and
the \VL{} returns true. Otherwise, some such update
has occurred, so the \VL{} returns \texttt{false}.

\myparagraph{The \SC{} Operation}
An \SC{} works in three stages:  (1) allocating the backup with the correct sequence number,
(2) trying to install the backup, and
(3) if successful and if the overwritten header was
a sequence number, it installs the cache.
In the first stage the algorithm determines the sequence number of the corresponding \LL{} (if the \LL{} performed the fast path then the expected tag is the sequence number, otherwise it reads the sequence number through the pointer).
It then allocates a backup node with a sequence number one greater (Line~\ref{lf:new}), and protects this pointer.  This protection is
 important in avoiding an ABA problem, since it prevents the location from being reused if another thread retires it.

In the second stage (lines~\ref{lf:start2}-\ref{lf:end2}), 
the algorithm attempts to install the backup with the \lstinline{cmp_xch} (Line~\ref{lf:cas1}).
If either the test for equality or the \lstinline{cmp_xch} on that line fail, it could be 
because another thread updated the header from a pointer 
to a sequence number without changing the sequence number (property                             
5).  In this case since the logical state has not changed, the
algorithm retries with a second \lstinline{cmp_xch} (Line~\ref{lf:cas2}) to change from the sequence number                            
to the pointer to the new backup. 
On the other hand, if the check on Line~\ref{lf:cas2} or the \lstinline{cmp_xch} fails, then the logical
state has changed and the \SC{} fails, freeing
the backup node and returning false.  

If the algorithm succeeds in installing the backup, and in doing so overwrites a pointer to a node,
it retires the overwritten node and returns true.   There is no attempt to install the cache in this
case.   On the other hand, if the installation succeeds and overwrites a sequence number, the
algorithm proceeds to the third stage (lines~\ref{lf:start3}-\ref{lf:end3}), which installs the
cache.
To do this, the algorithm copies the value into the                             
cache (line~\ref{lf:copywrite}) and attempts to change the header back to a sequence                            
number with a \lstinline{cmp_xch}---the same sequence number as in the backup node.
If successful, the \SC{} is done.  Note that this \lstinline{cmp_xch} does not change the logical                          
state since both val and \seqnumber{} remain the same.  If the \lstinline{cmp_xch}
fails, the algorithm tries again by using a protected load to load                          
the header (note that it must be a pointer by property 6), copying the                               
contents of the node into the cache, and retrying                          
the \lstinline{cmp_xch}.  This is repeated until success. 
When the \lstinline{cmp_xch}
succeeds, the sequence number is installed and the algorithm can                          
retire the backup that was overwritten by the \lstinline{cmp_xch} and return true.  

As with the wait-free algorithm the code as shown only supports a single outstanding \LL{}, but can be extended.

\vspace{-.05in}
\begin{theorem}
\label{proof:lf}
The big atomic object described in Algorithm~\ref{alg:lock-free-single-header} has linearizable \LL{} and \SC. 
Furthermore for $n$ big atomics each of size $k$, and $p$ processes, all operations that do not race take $O(k)$ time, operations are always lock-free, and the total memory usage is $nk + O(n + p^2k)$.
\end{theorem}
\vspace{-.1in}

\begin{proof}
It suffices to show that the implementation admits mutually consistent linearization points for every operation.

We first consider the correctness of \LL. There are two cases to consider: the fast path and the slow path. If the fast path succeeds (and hence the loads of lines~\ref{lf:ll:hdr1} and~\ref{lf:ll:hdr2} are the same sequence number), then by property~\ref{prop:noupdate} the state of the cache has not changed between the two loads, which further, by property~\ref{prop:logstate} is the logical state of the object. Thus in the fast path \LL{} returns the logical state of the object and may linearize anywhere in between the two loads of the sequence number. We now consider the slow path. If the check on line~\ref{lf:ll:ispointer} succeeds, then a backup pointer was loaded and protected. At this protected load, by property~\ref{prop:logstate}
, the logical state of the object was exactly the value pointed to by \texttt{hdr}, which is returned. Hence the \LL{} linearizes at the protected load.

We now consider the correctness of \VL. If the check on line~\ref{lf:validate:eqcheck} succeeds, then both the load on the previous line and the previous \LL{} read either the same sequence number or the same pointer and hence, by ABA freedom (property~\ref{prop:ll})
and property~\ref{prop:logstate}, no intervening update may have occurred. Thus the \VL{} linearizes at the load of the header on line~\ref{lf:validate:load} at which it observes this logical state. 

Instead, if the check on line~\ref{lf:validate:validatecheck} succeeds, then by property~\ref{prop:ll} the previous \LL{} linearized at the protected load of backup \texttt{p} on line~\ref{lf:ll:protect} pointing to a block whose seqnum is equal to that loaded at line~\ref{lf:validate:load}. In this case, by property~\ref{prop:ptrupdate} the cache was validated with the logical state pointed to by the \texttt{p} and no intervening update to the logical state may have occurred between the protected load of \texttt{p} and the load of the sequence number on line~\ref{lf:validate:load}. Hence the \VL{} linearizes there.

Finally, if neither check succeeds, and the \VL{} returns \texttt{false} on line line~\ref{lf:validate:fail}, then either:

\begin{enumerate}
  \item \texttt{actual\_tag} is a seqnum, and \texttt{expected\_tag} is a pointer.
  \item \texttt{actual\_tag} and \texttt{expected\_tag} are distinct pointers.
  \item \texttt{expected\_tag} is a pointer, and \texttt{actual\_tag} is a sequence number
        not equal to that associated with \texttt{actual\_tag}.
\end{enumerate}

In any of these cases some \SC{} must have linearized since the thread's prior \LL{};
we consider each in turn. First consider case (1); in this case, some \SC{} must have
installed a backup via a \lstinline{cmp_xch}. This \SC{} linearizes at this operation,
and hence a logical update has occurred since the prior \LL{} linearized. Case (2) is similar.
Finally, consider case (3). In this case, the \seqnumber{} has changed,
which by property~\ref{prop:seqincr} must have resulted from a logical update. In all
cases, \VL{} returns \texttt{false} iff a \SC{} linearized since the thread's prior \LL{},
and the \VL{} may linearize any point between the load on line~\ref{lf:validate:load}
and its return.


We now consider the correctness of \SC. Suppose the call returns on line~\ref{lf:end2}. In this case, then the CAS on line~\ref{lf:cas1} fails and the tag (read into \texttt{cur\_tag}) must have changed since the previous \LL. Furthermore, either \texttt{cur\_tag} is now a pointer, \texttt{cur\_tag} is now a sequence number not equal to that associated with \texttt{expected\_tag}, or \texttt{cur\_tag} is indeed equal to the sequence number equal to that associated with \texttt{expected\_tag}, but the CAS on line~\ref{lf:cas2} fails. We consider each of these cases in turn. If \texttt{cur\_tag} is a pointer not equal to \texttt{expected\_tag}, then by property~\ref{prop:ptrupdate} a logical update has occurred and the \SC{} linearizes when it learns of the update at line~\ref{lf:cas2}. Otherwise, if \texttt{cur\_tag} is a sequence number distinct from that associated with the \texttt{expected\_tag}, then again by property~\ref{prop:seqincr} another logical update must have occurred since the last \LL{} and the \SC{} linearizes at this update. Finally, if the CAS on line~\ref{lf:cas2} fails, then the sequence number previously observed at the load must have changed; by property~\ref{prop:seqincr}, another logical update must have occurred since the \LL{}, and the \SC{} linearizes when it learns of this
at line~\ref{lf:cas2}.

If the \SC{} does not return on line~\ref{lf:end2}, then it successfully installs the backup and linearizes at the \texttt{cmp\_xch} that installs it; either that at line~\ref{lf:cas1} or that at line~\ref{lf:cas2}. By property~\ref{prop:logstate}, this updates the logical state as desired; it remains to show that if one of these \texttt{cmp\_xch} operations succeeds then no \SC{} has occurred since the previous \LL{} to demonstrate that this corresponds to a valid logical update. In either case, the successful \texttt{cmp\_xch} installs a backup with a sequence number one greater than that previously installed, preserving property~\ref{prop:seqincr}. Further, this backup is fresh, preserving property~\ref{prop:freshptr}.

Now case on which \texttt{cmp\_xch} succeeds. First suppose it succeeds at the first \texttt{cmp\_xch} on line~\ref{lf:cas1}. There are two further subcases to consider: if the \texttt{cmp\_xch} swaps out a pointer or a sequence number. If it is a sequence number, then no logical update (and hence no \SC) has occurred since the previous \LL{}, as by property~\ref{prop:seqincr} every logical update increments the sequence number. This preserves property~\ref{prop:sequpdate}  by changing the sequence number to a pointer. Now suppose that a pointer was swapped out, performing a logical update and preserving property~\ref{prop:ptrupdate}. Because the expected pointer is protected by property~\ref{prop:ll} and only freshly allocated backups are installed by property~\ref{prop:freshptr}, no logical update could have occurred since the previous \LL{} (as this would require the pointer to have been recycled).

Now suppose the second \texttt{cmp\_xch} on line~\ref{lf:cas2} succeeds after the first fails; in this case, the actual tag encountered stored in the header word is the sequence number of the backup pointer recorded in \texttt{expected\_tag}---the backup read during the previous \LL{} was validated in the meantime and replaced with a sequence number. Because the sequence number has not changed, it follows from property~\ref{prop:seqincr} that no logical update has occurred since the previous \LL{} as desired. Because the sequence number is replaced by a new backup pointer, property~\ref{prop:sequpdate} is also preserved in this case.

We now consider the correctness of stage 3; the validation of the cache. We maintain the loop invariants that the cache contains the value of \texttt{new\_backup}, that \texttt{new\_seqnum} is the sequence number associated with \texttt{new\_backup}, and finally that \texttt{new\_backup} is protected. Initially, these invariants are truthified by storing \texttt{v} (the argument to the \SC) to the cache and having previously protected \texttt{new\_backup} on line~\ref{lf:sc:backup_protect}. Furthermore these invariants are obviously maintained by the loop, and when it terminates (via success of the \texttt{cmp\_xch} to install \texttt{new\_seqnum}) the cache therefore contains the value of the logical version associated with \texttt{new\_seqnum}. Note that this validation loop preserves property~\ref{prop:noupdate} since by property~\ref{prop:ptrtoseq} no other thread could have stored a sequence number to the header. Furthermore, by property~\ref{prop:exclown}, there are no write-write races to the cache during this validation loop.

The algorithm is lock-free since the \LL{} and \SC{} can only loop if there is a change to the logical state.
\VL{} never loops at all. 
For space, note that there is only a backup associated with a bigatomic location if there is an ongoing write
to that location, of which there may be at most $p$. Hence the space incurred by the backups is $pk$.

Additionally, \came{} incurs the same space terms for
the $n$ objects themselves and memory reclamation as in \cawf{},
incurring an additional $nk + O(n + p^2 k)$. The $pk$ term
incurred by the backups is thus absorbed, resulting in
$nk + O(n + p^2 k)$ space total.

\end{proof}

    \section{Inlined Hash Table with Big Atomics}
\label{sec:hashtables}

Here we briefly describe \parlayhash{}.  \parlayhash{} supports growing and arbitrary types, but here we only outline
the part that is relevant to big atomics.  More information on the hash table, and
the code itself can be found on its website~\cite{parlayhash}.
The design of \parlayhash{} was motivated by the desire to minimize the number
of cache misses required by each operation.
For sequential hash tables, it is well known that using open addressing to store key-value
pairs directly in the hash-array
is an effective way to reduce cache misses~\cite{BenderKK21,richter2015seven}.  
This is particularly true with linear-probing since adjacent probes are likely to be on the same cache line,
and typically a lookup can involve just a single cache miss.
Google's Swiss Table~\cite{swisstable} and Facebook's F14~\cite{follyF14}, for example,
take advantage of this to achieve state-of-the-art sequential hash table performance.

The problem with applying this approach concurrently is that loading and updating an entry
can require updating multiple words atomically.   In addition to having a key and value in
each slot, these approaches also have a control byte to indicate the state of each bucket (e.g., full, empty, tombstoned).
All these have to be updated atomically.  
This is not hard to do with a lock, but more challenging to do lock-free.
Furthermore, if a probe sequence visits multiple slots, all
might need to be considered atomically for correctness.  

The idea of \parlayhash{} is to have each bucket store a group of entries as a single big atomic.
The size of a group is picked to fill a cache line (if at least one fits).   For example, for 8-byte keys and values, three entries fit in a 64-byte cache line, accounting for the bucket and big-atomic headers, which are 8 bytes each.
A table lookup consists of hashing the
key to identify the correct bucket, and then loading and searching within the bucket. Similarly, an update hashes
to a bucket and atomically updates it with an \SC{}.   If a bucket runs out of slots, then new entries are added to
the bucket using separate chaining.  The bucket header keeps the number of filled slots and a pointer
to the chain.   The approach avoids any need for atomic synchronization across separate buckets.  A lookup always requires a single big atomic load, and an update requires an \LL{} and an \SC{} of the updated state, and perhaps a retry if the \SC{} fails due to a concurrent update.
Any big atomic supporting \LL{}/\SC{} can be used.  If a wait-free or lock-free version is used, then all operations
are lock-free when the table is not growing.  Some steps in a grow-cycle involve fine-grained locks.

    \section{Other Applications of Big Atomics}
\label{sec:applications}

In this section we describe several applications of big atomics.
Some of these are direct applications: generally, algorithms that, as written, require big atomics in order to implement.
Others are potential applications: generally, opportunities to simplify ad hoc solutions and complex code.

\paragraph{Complex concurrent data structures} 
In concurrent data structures it is often necessary to atomically update multiple fields.
An example is for the widely used Natarajan trees~\cite{Natarajan14}.    
In addition to using flags, this data structure has a complicated mechanism of flagging and tagging due to the fact it needs to atomically update one of two child pointers and whether one or both children are being deleted.    
It does this by using an intricate mechanism that first flags one of the child pointers, and then tags the other (two separate bits on the pointers).     
With a big atomic, all this can be avoided, and marking a node for deletion becomes almost trivial.    
In particular the two child pointers and a flag indicating one of \{undeleted, left-deleted or right-deleted\}, could be placed in a big atomic.   
Marking an internal node for deletion then simply involves a CAS on the big atomic setting it to left-deleted or right-deleted.    
As usual with algorithms that mark nodes for deletion, the operation will then shortcut the node out, or if it is delayed, another operation will see the deletion flag and do the shortcut for it shortcutting from the parent to the right or left child depending on whether it is right or left deleted.   
Updating the left or right pointer also involves a CAS on the big atomic.   
Importantly setting the delete bit prevents any other update from concurrently changing the right or left child of the deleted node since their CAS will fail.

Jayanti and Tarjan develop an efficient concurrent union-find data structures~\cite{JTB19,JT21,JPhD23}.   
It requires updating three fields atomically.
Other data structures include a snapshot algorithm~\cite{BCW24}, which requires several applications of big-atomics on up to five fields, a wait-free fetch-and-increment from CAS~\cite{EW13}, requiring a 6-field big-atomic, a strongly linearizable LL/SC from CAS~\cite{NW24}, requiring 4-field big atomics, a DCAS (CAS on two unrelated location)~\cite{BGW24}, requiring a 3-field big atomic, and bounded time stamps~\cite{BJW23} requiring a 4-field big atomic.

\paragraph{Transactional memory}
In software transactional memory systems, multiple fields often need to be updated atomically.   
For example in DSTM~\cite{HLMS03} it is necessary to atomically update a pointer to a transaction descriptor, an old object pointer, and a new object pointer.   
Without an big atomic these systems require a level of indirection for which references to an object must go through.

\paragraph{Atomic markable pointers} 
Many lock-free concurrent algorithms store various mark bits on pointers~\cite{harris2001pragmatic,Natarajan14}.
Although this is easy in languages such as C++ where raw pointers can be manipulated as values, it is not possible in most garbage collected languages, e.g. Java, since the language does not expose pointer types to the user, i.e., pointers are only used internally by the runtime (and the garbage collector already uses various bits on pointers).
To use mark bits in these language without breaking language abstractions requires a level of indirection with a pointer to a pair containing the actual pointer and the required flag(s), e.g, Java's \texttt{AtomicMarkableReference}.
This effectively is exactly a two-word atomic---the goal is to update and read the pointer and flags atomically.   
As far as we know current implementations do not use a big-atomic abstraction for this, however.    
The problem is that it would either be locking (which is exactly what these are trying to avoid) or would require an extra level of indirection anyway.   
However, an efficient lock-free big atomic that avoids indirection would solve the problem, avoiding the indirection while being lock-free.   

\paragraph{Multiversioning}  
There are many applications of big atomics that appear in multiversion concurrency control~\cite{reed78,BG83,Postgres12,SQL13,neumann2015fast,Wu17,LKA17} and version lists \cite{WeiBBFRS21,JayantiJayantiVersionedSnapshot}.
For database systems one needs to store for each version of a value, the value at that version, a time stamp, and a next version pointer.  
Using big-atomics allows the first version, which is most commonly accessed, to be stored inline and updated atomically.  
This saves a level of indirection as long as the big-atomics themselves do not introduce a level-of-indirection.  
The idea of muliversioning in version lists is to maintain a history of the values at a location.   
This in turn can be used to construct a valid snapshot of the state of the data structure (i.e. a reader will read the values that were installed at a particular time). 
A version list contains a linked list of versions, each which contains a timestamp, the value written at that timestamp, and a next pointer to the next version.    
Most reads will only look at the most recent version, but some will have to traverse the list to find the correct version for their timestamp.
The problem with multiversioning is that reading a value requires reading through a head pointer to the first link even if reading the first version.   This can significantly degrade performance.    
One could place the first link in the object itself instead of the head pointer, but then all of its fields have to updated atomically.   
This is a direct implementation of big atomics.  Indeed some multiversioning systems do this, but require locks.   
Using an indirect implementation of big atomics does not help since it has no improvement over just using a head pointer.   
A lock-free direct implementation of big atomics yields lock-freedom and avoids the indirection.

    \section{Benchmarks}
\label{sec:benchmarks}

We have implemented our big atomics as a C++ library and evaluated them on several sets of benchmarks, one set comparing their performance directly in a simple microbenchmark, and another in which the big atomics are used to implement a hashtable using the strategy described in Section~\ref{sec:hashtables}. We run
our experiments on a \texttt{c8i.metal-48xl} Amazon EC2 instance, which is equipped with an Intel Xeon 6, 6975P-C  
processor with 96 cores (192 SMT threads), 12-channel DDR5-6400 memory, and 480MB of L3 cache. All experiment code is written in C++ and compiled with GCC 13 with optimization level \texttt{O3}. 
We turn off turbo boost since leaving it on gives erratic performance.  



\subsection{Big atomic microbenchmarks}\label{sec:hasharray-bench}
We first evaluate the performance of our big atomics against several baselines. The baseline implementations include (1) a SeqLock (\emph{SeqLock}), (2) an atomic pointer to a heap-allocated value (\emph{Indirect}), (3) an external lock pool (\emph{SimpLock}), and (4) \lstinline{std::atomic}. The \lstinline{std::atomic} used by GCC compiles to hardware atomic operations for up to two words,
and uses the GNU libatomic library on Linux for larger sizes, which uses a small number of shared locks.  We benchmark our implementations of the wait-free big atomic from Section~\ref{sec:wait-free-load-cas} (\emph{\cawf{}}) and our lock-free memory-efficient big atomic from Section~\ref{sec:lock-free-load-cas} (\emph{\came{}}).

Our microbenchmark implements a map from integers $(0,\ldots,n-1)$ to values.   It is implemented as an array of elements, each a big atomic containing a full/empty flag, and a value (array of integers).   A find does an atomic load of the element, and returns the value if full.  An insert does a load, and if empty does a CAS to try to swap in the value.  A delete does a load, and if full tries to empty with a CAS.
We measure the number of operations executed across all threads in billions of operations per second (Bop/s). The parameter space of our benchmarks ranges over $p$: the number of threads, $n$: the number of big atomics, $u$: the percentage of updates (equal mix of inserts and deletes), $z$: a Zipfian \cite{YCSB} parameter controlling the contention (chosen between $0$ = uniformly random, and $1$ = every operation always selects the same big atomic), and $k$: the size of the element held by the big atomic (flag+value) measured in 64-bit machine words. 
We align the elements at 64-byte boundaries so even 1-word values do not fit in cache at $n = 10^7$.

Figure~\ref{fig:c7i-hasharray} illustrates eight benchmarks.
Each corresponds to varying one parameter while holding the others fixed.  
Six correspond to varying $u$, $z$, and $n$ each for two values of $p$, one without oversubscription ($p = 192$), and one with oversubscription ($p=4 \times 192 = 768$).     One corresponds to varying $k$ and one corresponds
to varying $p$.
The defaults are $n =10^7$, $z=.9$, $u = 5$\%, and $k=4$ (i.e., 32 bytes).
Results for the geometric mean of performance across parameters were given in Figure~\ref{fig:expsummary}.

\begin{figure}

    \begin{subfigure}{\columnwidth}
      \centering
  		\includegraphics[width=0.95\columnwidth]{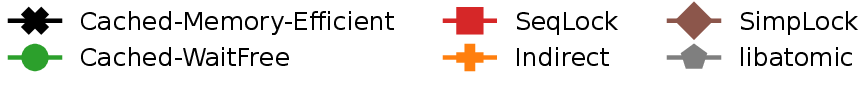}
    \end{subfigure}

    \medskip

        \includegraphics[height=3.4cm]{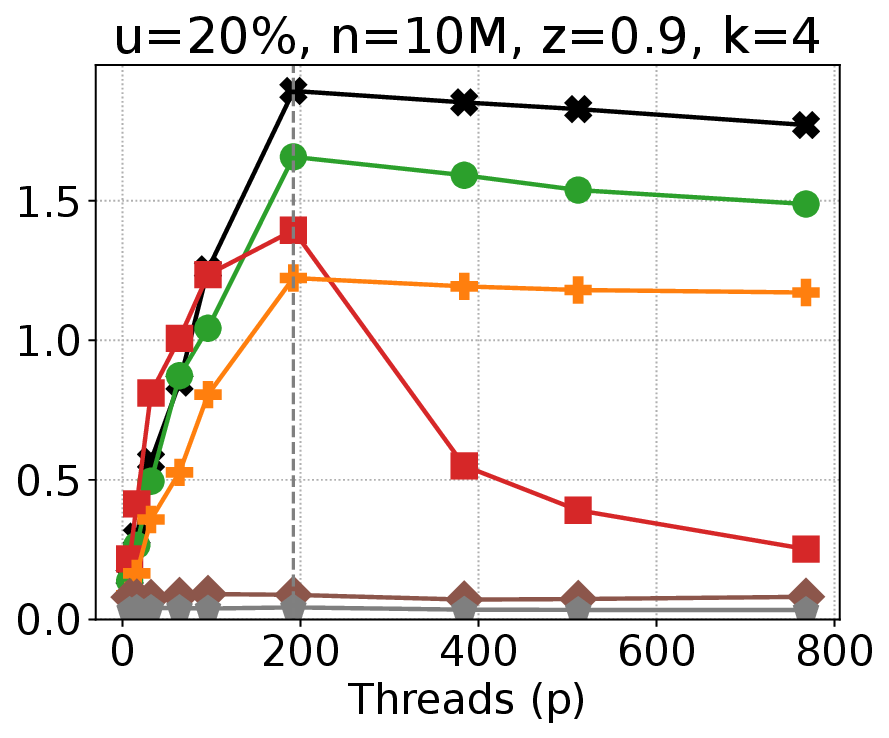}
        \includegraphics[height=3.4cm]{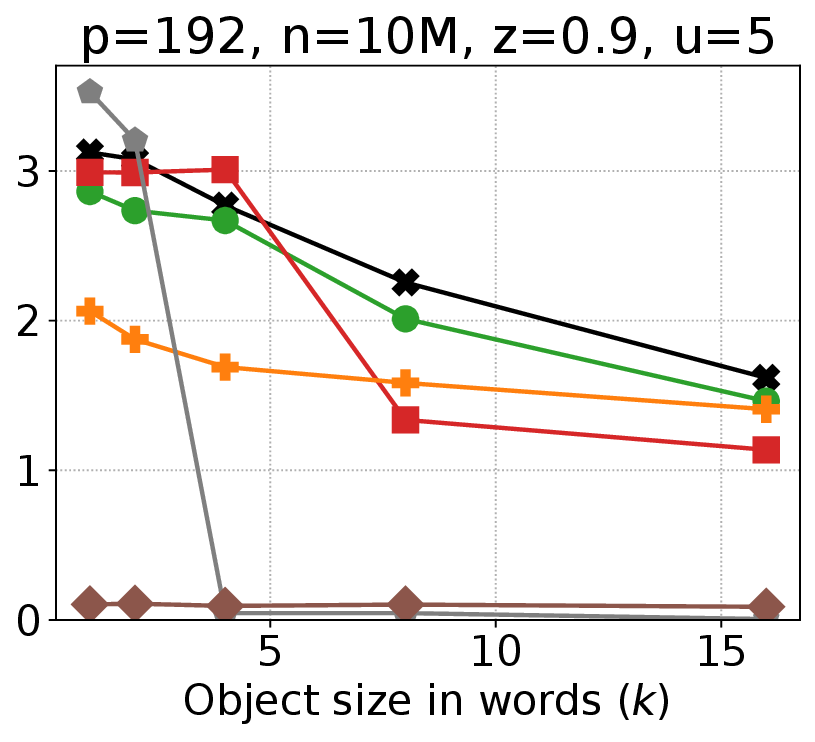}
        
        \includegraphics[height=3.4cm]{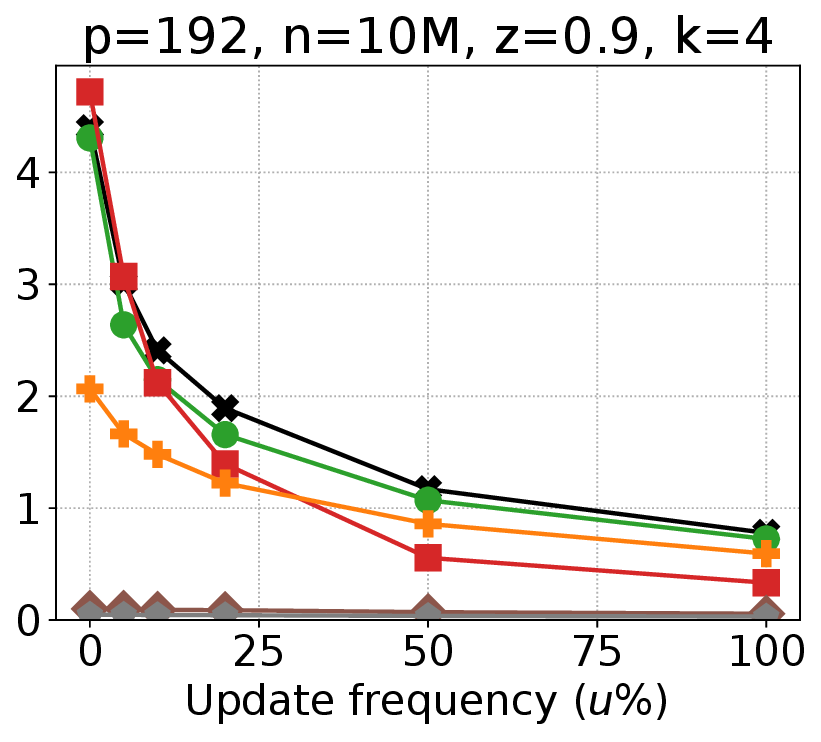}
        \includegraphics[height=3.4cm]{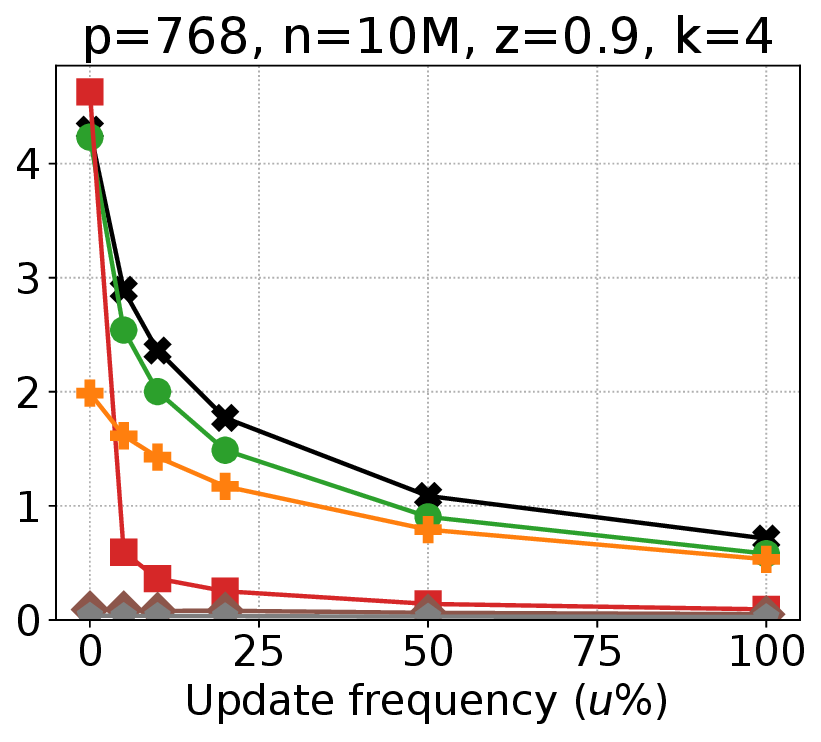}
        
        \includegraphics[height=3.4cm]{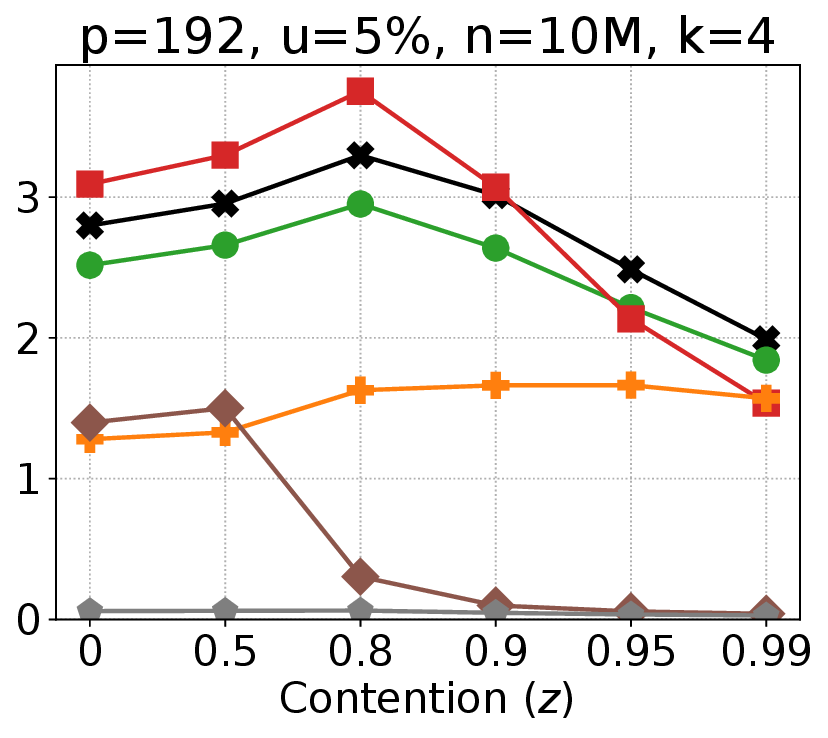}
        \includegraphics[height=3.4cm]{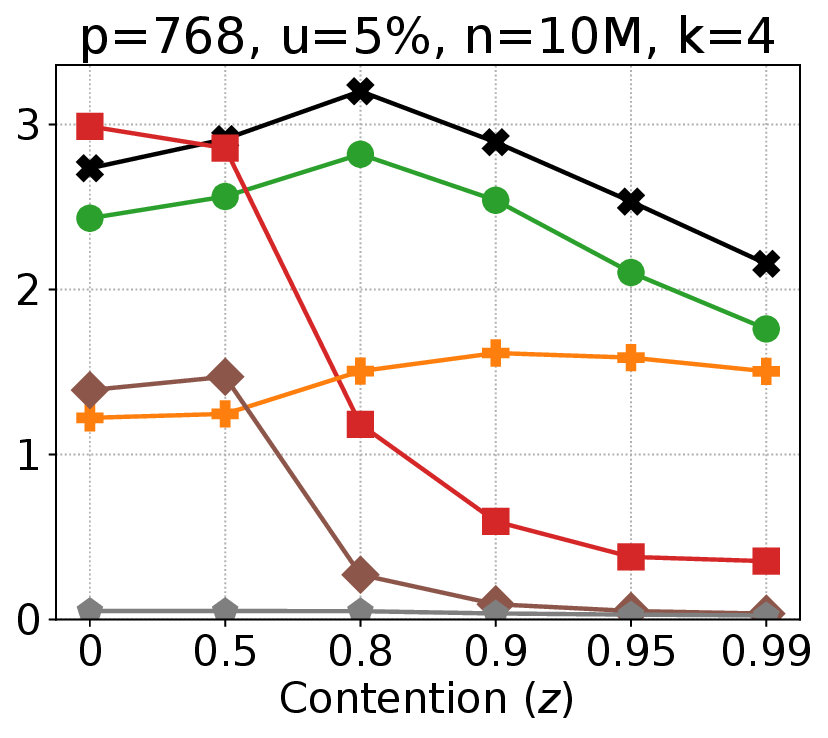}
        
        \includegraphics[height=3.4cm]{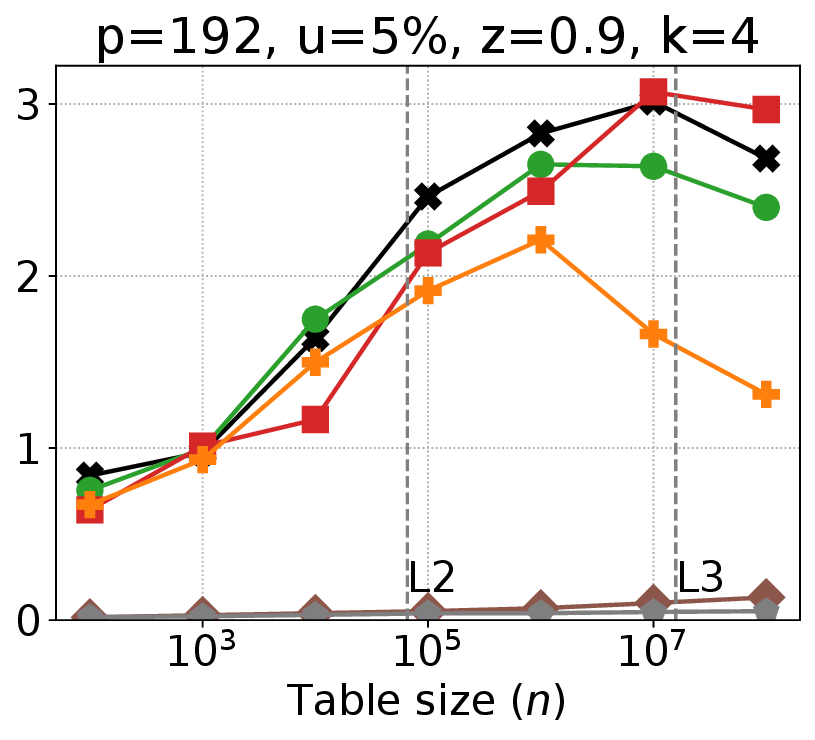}
        \includegraphics[height=3.4cm]{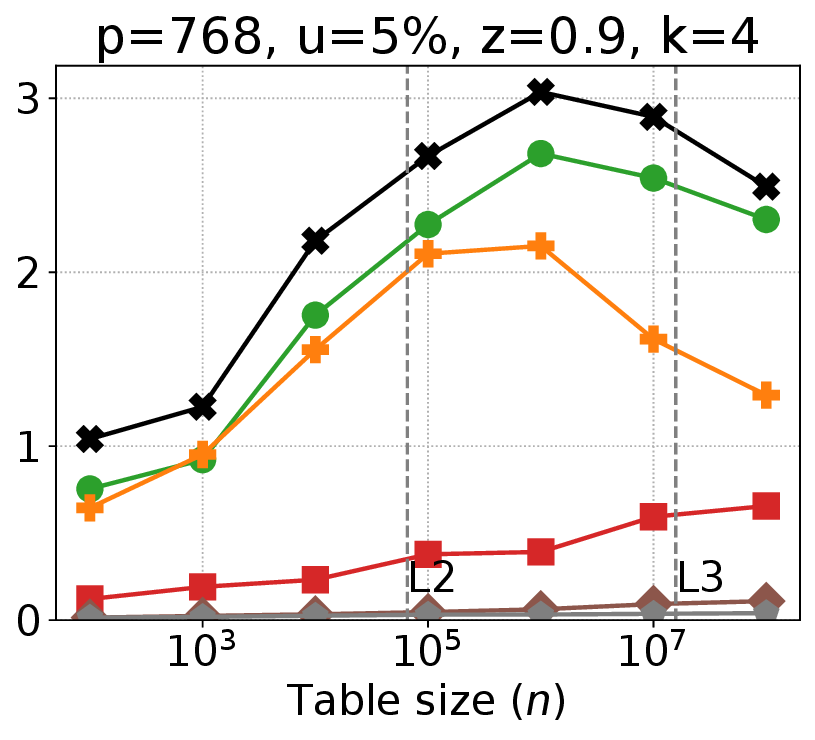}

    \caption{Throughput in billions of operations per second for big atomic implementations across varying thread counts $(p)$, update frequencies $(u)$, contention parameters $(z)$, table sizes $(n)$, and element size measured in number of words $(k)$.\vspace{-.1in}}
    \label{fig:c7i-hasharray}
\end{figure}

\myparagraph{Varying $u$}  The second row of Figure~\ref{fig:c7i-hasharray} varies $u$, the update rate, from $0\%$ (read only) to $100\%$ (update only).    We run this both with oversubscription ($p = 768$) and without ($p=192$).   \came{} performs best except at very low update rates, where \seqlock{} performs slightly better.   The \seqlock{} performs badly at higher update rates because the contention caused by having $z = .9$ (see below for varying $z$).
The \indirect{} algorithm is about twice as slow as the others with modest update rates ($\leq 20\%$), as expected since each access involves two cache misses instead of one.   With no updates ($u = 0$) \indirect{} is actually more than twice as slow due to the additional cost of safe memory reclamation.  At higher update rates, the cost becomes dominated by the compare-and-exchange, which all algorithms need, and therefore \indirect{} performs comparably with the others.

The \simplock{} performs badly across all update rates since even loads block each other, and $z = .9$ means the contention is
reasonably high.
The libatomic implementation similarly performs badly across the whole range.
With oversubscription, all non-blocking algorithms outperform \seqlock{} by a large factor, except with $u = 0$ (where no locks are taken).   
As with other lock-based algorithms, this is due to threads being descheduled while holding the lock and blocking other threads from making progress, or creating additional overhead for thread switches. 

\myparagraph{Varying $z$}  The third row of Figure~\ref{fig:c7i-hasharray} controls $z$, the contention, from $0$ (uniform) to $0.99$ (very high contention).   
As $z$ increases, contention becomes high, but so does the number of machine cache hits, which can therefore sometimes counter-intuitively improve performance.  Without oversubscription and at low contention \seqlock{} slightly outperforms our cached structures
(within 10\% or so).     At higher contention, however, our cache-based implementations do better.    This appears to be because
the \seqlock{} effectively holds a lock during the full update blocking both readers and writers from progressing.  On a contended location this can be problematic.   Our cache-based algorithms, and \indirect, effectively only have to succeed on a CAS to perform an update.   On a contended location the big-atomic cache will not get refilled on every update and most loads will find their value indirectly in the backup. Furthermore, since it is a hot location, the backup node is likely to be in the memory cache reducing the
cost of the indirect load although still requiring a fence for safe memory reclamation.
We note that at the highest zipfian, even \indirect{} is (slightly) faster than \seqlock{}.  

\simplock{} does OK at low contention, but does terribly at higher contention since even the loads block each other.
With oversubscription, the lock-free algorithms are hardly affected, and the lock-based ones are dramatically affected.   Although \seqlock{} is about 5\% faster than \came{} with a uniform distribution (i.e., almost no contention), it has a fraction of the performance at high contention.

\myparagraph{Varying $n$}  The last row of Figure~\ref{fig:c7i-hasharray} varies $n$, the table size. 
We note that at small $n$ the contention becomes higher, causing lower performance, but at large $n$ the data no longer fits in cache, also causing lower performance.  
We have marked
on the graphs the approximate points at which the data no longer fits within the level 2 (L2) and level 3 (L3) cache of the machine.   We also note that overhead due to various tests and non-cache-miss causing instructions becomes more important at smaller sizes where the cost is not dominated by memory traffic.    At $n=1000$ all perform about equally well.   At the largest size ($n = 10^8$) \seqlock{} somewhat outperforms the others.   \indirect{} suffers the most at larger sizes, presumably because the cost of indirection becomes more important as the data stops fitting in the cache.
With oversubscription, again, the non-blocking algorithms perform as well as without oversubscription, but \seqlock{} suffers substantially.

\myparagraph{Varying $k$} The top right graph of Figure~\ref{fig:c7i-hasharray} varies the element size from $1$ word ($8$ bytes)
to $16$ words ($128$ bytes).   Unsurprisingly, here libatomic takes its first victory at $k=1$ since the operations compile into machine instructions (a single-word read or a single-word compare-exchange). Interestingly, even in this case, libatomic only slightly outperforms the \seqlock{} and our algorithms.
At $k=2$, libatomic compiles to a double-word compare-and-exchange, and still very slightly outperforms \came{} and \seqlock{}.
Beyond two words, libatomic retakes its place at dead last.  All other methods scale reasonably well to sizes that require multiple cache lines.  \seqlock{} seems to suffer slightly more beyond one cache line, presumably because it has to hold a lock for longer.  We note that for 16 words (128 bytes) all the implementations require touching at least three cache lines---two for the data, and at least one for the meta-data (version count, backup pointer, etc.).   The cost of the methods seems to converge at larger sizes since the cost is dominated by accessing the data itself.

\myparagraph{Varying $p$}  The final benchmark (top left in Figure~\ref{fig:c7i-hasharray}) controls $p$, the number of threads, scaling from $1$ to $768$, which is $4 \times$ the number of SMT threads of the machine.  A vertical line depicts the number of threads (192), beyond which is at oversubscription.   The experiments show that even at mild oversubscription, the lock-based algorithms suffer.   The lock-free algorithms are mostly flat with oversubscription.     

\begin{figure}
    \begin{subfigure}{\columnwidth}
      \centering
  		\includegraphics[width=0.99\columnwidth]{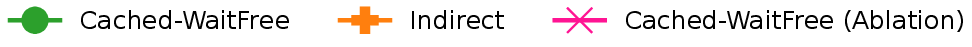}
    \end{subfigure}

    \medskip

    \begin{minipage}[c]{0.45\columnwidth}
    \includegraphics[height=3.4cm]
        {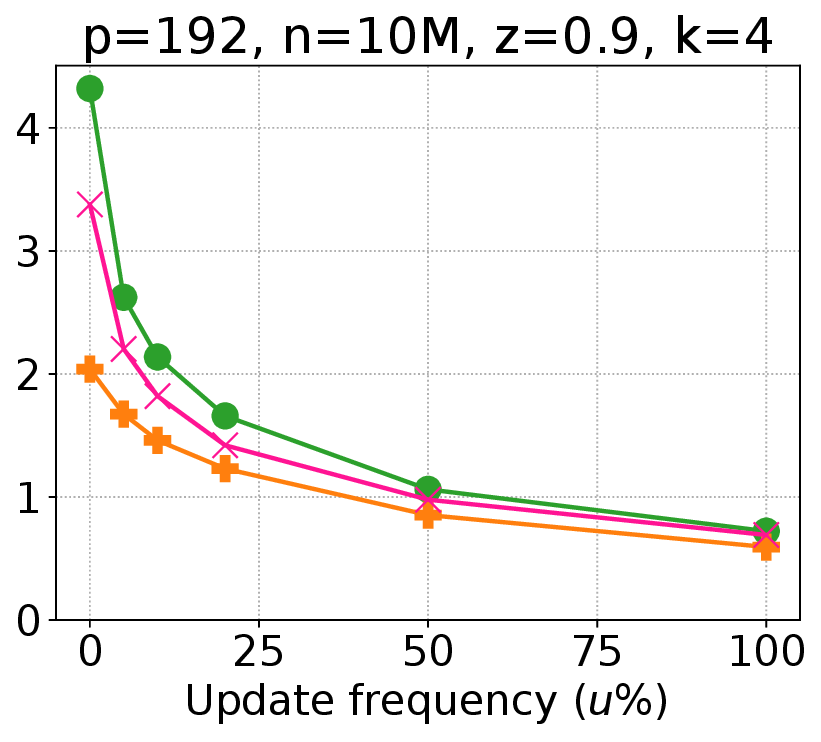}
    \end{minipage}
    \hfill
    \begin{minipage}[c]{0.54\columnwidth}
        \caption{Throughput in billions of operations per second for \indirect{} and \cawf{} compared against a version of \cawf{} that always announces a hazard pointer, for various update frequencies $(u)$.\vspace{-.15in}}
        \label{fig:waitfree-ablation}
        \bigskip\bigskip
    \end{minipage}
\end{figure}

\myparagraph{Cost of Hazard Pointers in the Fast Path}
Our fast-path algorithms improve over the performance of \indirect{} because of \emph{two} optimizations. On the fast path, they avoid an extra memory load/cache miss by not needing to follow the backup pointer, but they also avoid the need to announce a Hazard Pointer. 
The result is that \cawf{} is between 1.5-2.5x better throughput than \indirect{} depending on the workload (higher update rate and contention shrinks the gap since the slow path needs to be taken more frequently, and costs become dominated by updates, which do not have a fast path).

In theory, both of these optimizations should contribute to the performance improvement. Avoiding dereferencing the backup pointer saves memory bandwidth and latency, and avoiding the hazard pointer announcement avoids an additional memory barrier since the announcement has to be fenced. To distinguish the effect of the two optimizations, we perform an ablation experiment where we always announce a Hazard Pointer even in the fast path. The result is depicted in Figure~\ref{fig:waitfree-ablation} where we 
compare \cawf{}, which neither reads indirectly nor announces on the fast path, a version of \cawf{} that announces on the
fast path, but does not read indirectly, and \indirect{}, which both announces and reads indirectly always.

For a read-only workload (meaning the fast path is always taken), \cawf{} is 2.1x better throughput than \indirect{}. Avoiding the Hazard Pointer announcement improves throughput by 28\% over not doing so (Ablation), which means that this optimization accounts for 41\% of the improvement. The gap between \cawf{} and \indirect{} shrinks as updates become more frequent (the slow path is taken more often), but the effect of avoiding the Hazard Pointer remains significant; at 20\% updates, \cawf{} is 35\% better throughput than \indirect{}, about 53\% of which is attributable to avoiding the Hazard Pointer, and at 50\% updates, \cawf{} is 25\% better than \indirect{}, about 40\% of which is attributable to avoiding the Hazard Pointer. We therefore see that regardless of update rate, avoiding the Hazard Pointer announcement consistently contributes 40--50\% of the improved throughput.

\myparagraph{Hardware Transactional Memory (HTM)}
While Intel discontinued hardware transactional memory in 2021, due to security bugs \cite{htmdead}, we include some experiments on an older machine which supports Restricted Transactional Memory (RTM), a form of HTM. The machine is equipped with four Intel Xeon E7-8867 v4 totalling 72 cores (144 SMT threads), each with four-channel DDR4-2400 memory (totalling 1TB), and 45MB of L3 cache (per socket). Benchmarks were executed with \emph{numactl -i all} to interleave memory allocations across the sockets.

Our implementation tries to perform the operation using a hardware transaction ten times before falling back to a spinlock, since RTM in general is not guaranteed to ever succeed. In Figure~\ref{fig:aware-hasharray}, we show the results of varying the thread count $p$, the contention parameter $z$, the update frequency $u$, and the table size $n$, using our microbenchmark. 
With low contention, HTM scales well with the number of threads, but still lags behind \came{} at all thread counts. However, we see that as contention increases, HTM's performance degrades drastically. In particular, as $z$ surpasses $0.8$ its performance plummets while SeqLock and \came{} continue to perform well or even improve. Interestingly, for HTM, the update frequency ($u$) barely affects the performance, while most other methods lose performance as $u$ increases. Nonetheless, HTM still can not outperform the other methods even at the highest value of $u$ (100\% updates). As we vary the table size $n$, HTM is drastically outperformed as long as the data fits in cache, with a smaller gap once the data is too big to fit in cache.

We note that on this older machine \seqlock{} does better compared to \came{} than on the more modern machine used for our other experiments.

\begin{figure}

    \begin{subfigure}{\columnwidth}
      \centering
  		\includegraphics[width=\columnwidth]{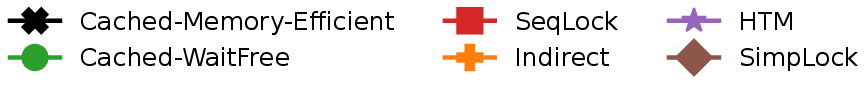}
    \end{subfigure}

    \medskip

    \begin{subfigure}{0.99\columnwidth}
        \centering
        \includegraphics[height=3.4cm]{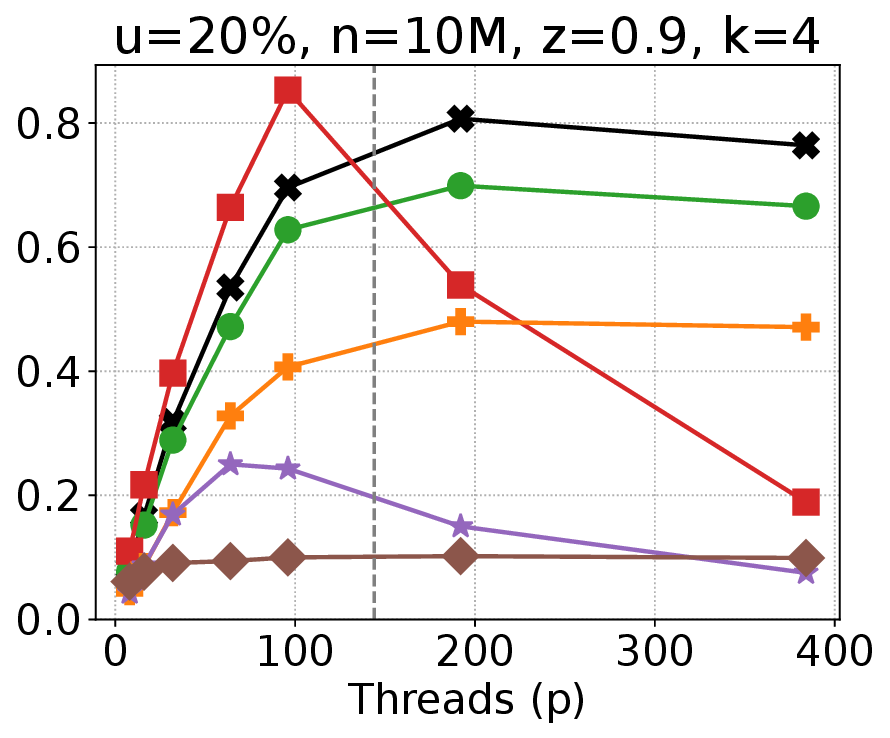}
        \includegraphics[height=3.4cm]{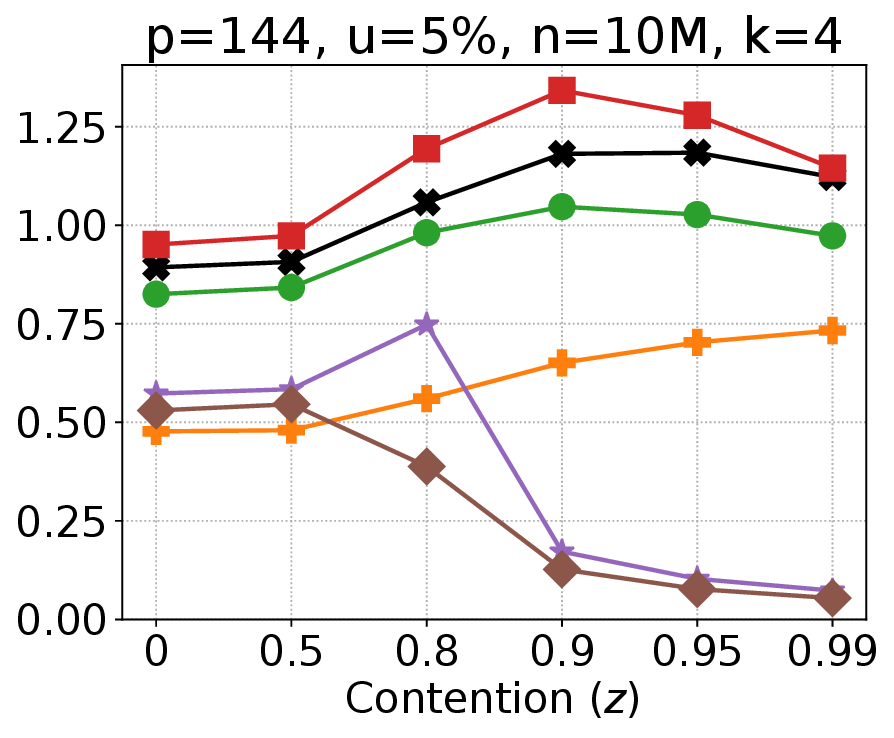}
        
        \includegraphics[height=3.4cm]{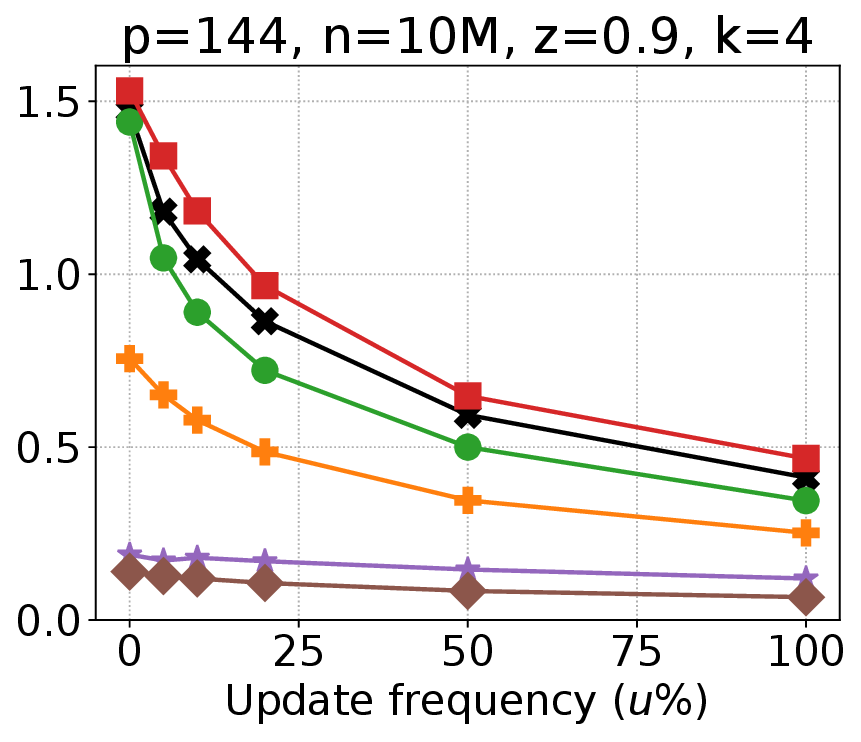}
        \includegraphics[height=3.4cm]{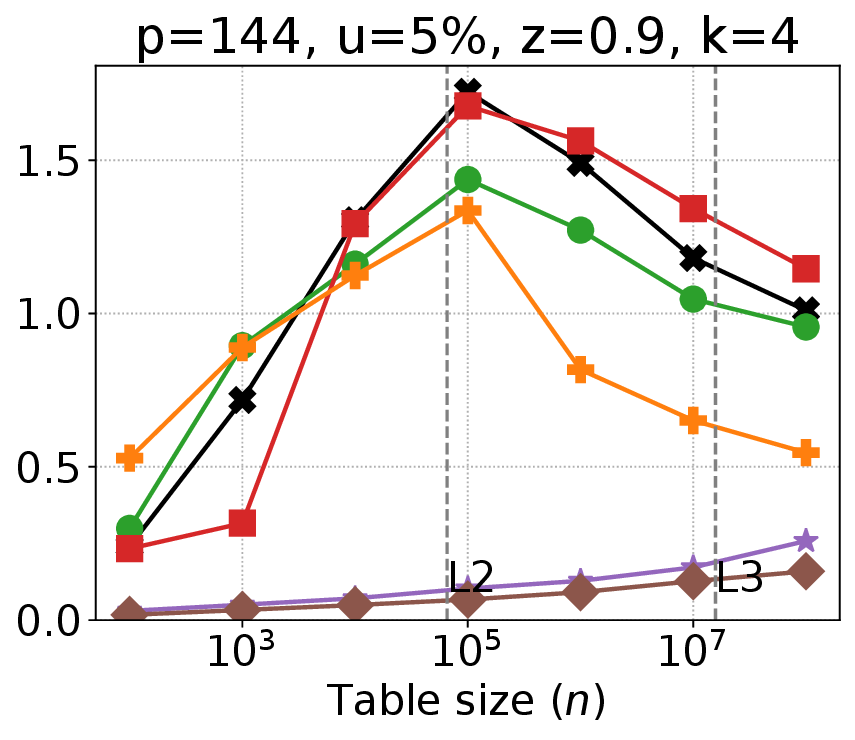}
    \end{subfigure}
    \caption{Throughput in billions of operations per second for various big atomic implementations including hardware transaction memory (HTM) on an older four-socket machine across varying thread counts $(p)$, contention $(z)$, update frequency $(u)$, and table size $(n)$.
    }
    \label{fig:aware-hasharray}
\end{figure}

\myparagraph{Memory Usage}
The memory used by all implementations is limited to the big-atomic structures themselves and, for those that have indirect values, the memory for the indirect structures.    The bounds given in Table~\ref{tab:properties} therefore translate to the practical implementation as well, although the constants in the big-O need to be filled in.
Here we specify them for our implementation. \indirect{} uses $n(k + 1) + c_h p^2 k$, \simplock{} uses $n(k + 1)$, \seqlock{} uses $n(k+1)$, \cawf{} uses $2n(k + 2) + c_h p^2 k$, and \came{} uses $n(k + 2) + c_h p^2 k$.   Here $c_h$ is a parameter used in the hazard pointer collector and can be tuned (raising it will speed up the implementation due to less frequent collection). Experimentally, a value for $c_h$ that makes the cost of collection negligible is one such that there are around one thousand garbage elements per thread.   Therefore one can consider the $c_h p^2 k$ term as  $\leq p * 1000$.  For modestly large $n$ this is a low-order term.

\subsection{\ourHashtable{} Benchmarks}\label{sec:hashlist-bench}

Our next benchmark compares the performance of our hashtable using big atomics (\ourHashtable{}), as described in Section~\ref{sec:hashtables}. 
In particular, we compare the performance of \ourHashtable{} using the same implementations of big atomics
as in the microbenchmark, excluding libatomic and \simplock{} since they were strictly worse.
Our experiments  do not use any growing.

The benchmark loop is similar to our first benchmark set.  We create a hashtable of size $n$ storing eight-byte keys and eight-byte values, then each of $p$ threads repeatedly chooses a random key from a Zipfian distribution with parameter $z$ and randomly performs a find, insert, or delete. Insert and delete are performed $u\%$ of the time, choosing uniformly between the two, with find happening otherwise $(100-u)\%$ of the time.   We report results using a cross section of parameters.
We use a load factor of one, with the size rounded to the next power of two.  The results are shown in Figure~\ref{fig:c7i-hash-list}.
Overall the performance using the different big-atomic strategies looks similar to the microbenchmarks.  The differences among the algorithms are not quite as large, however, since there are other overheads in the hash tables.

\begin{figure}
    \begin{subfigure}{\columnwidth}
      \centering
  		\includegraphics[width=0.72\columnwidth]{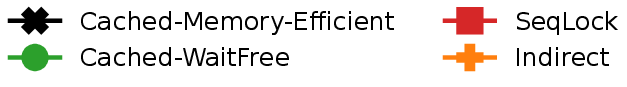}
    \end{subfigure}

    \medskip

    \includegraphics[height=3.4cm]{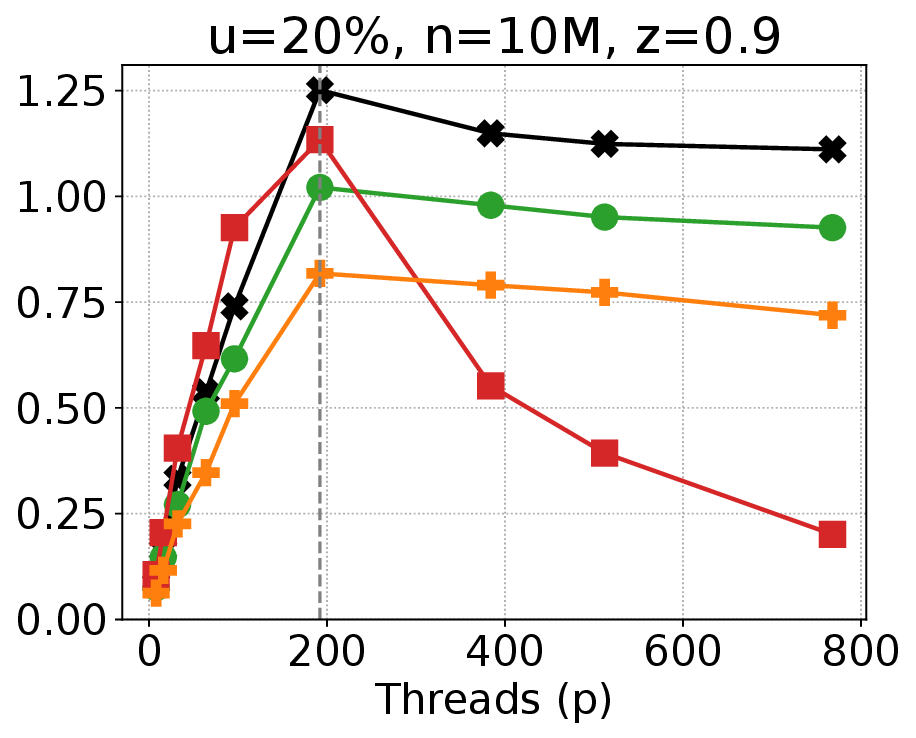}
    \includegraphics[height=3.4cm]{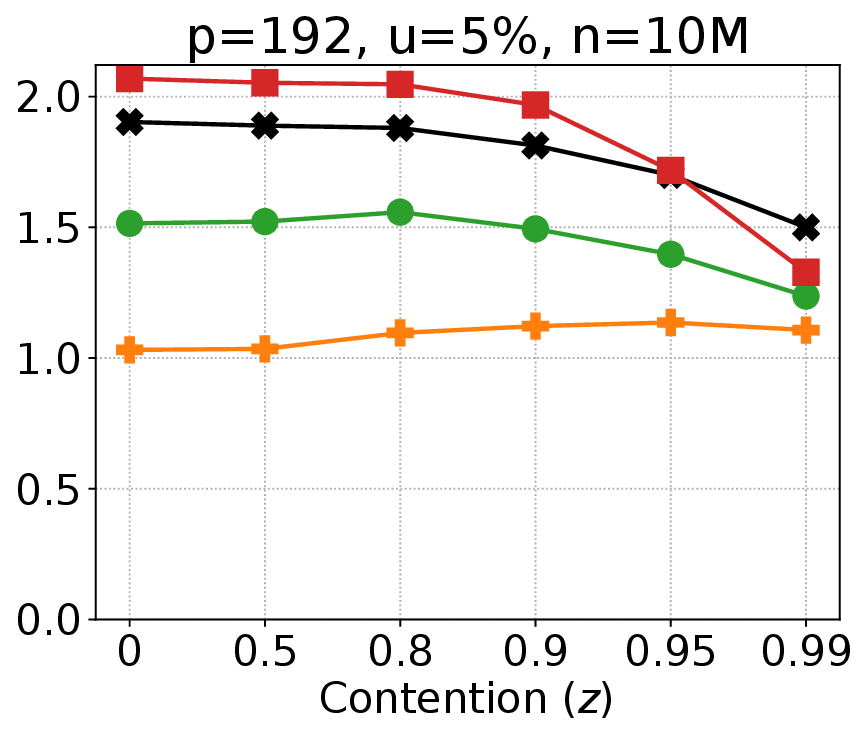}
    \includegraphics[height=3.4cm]{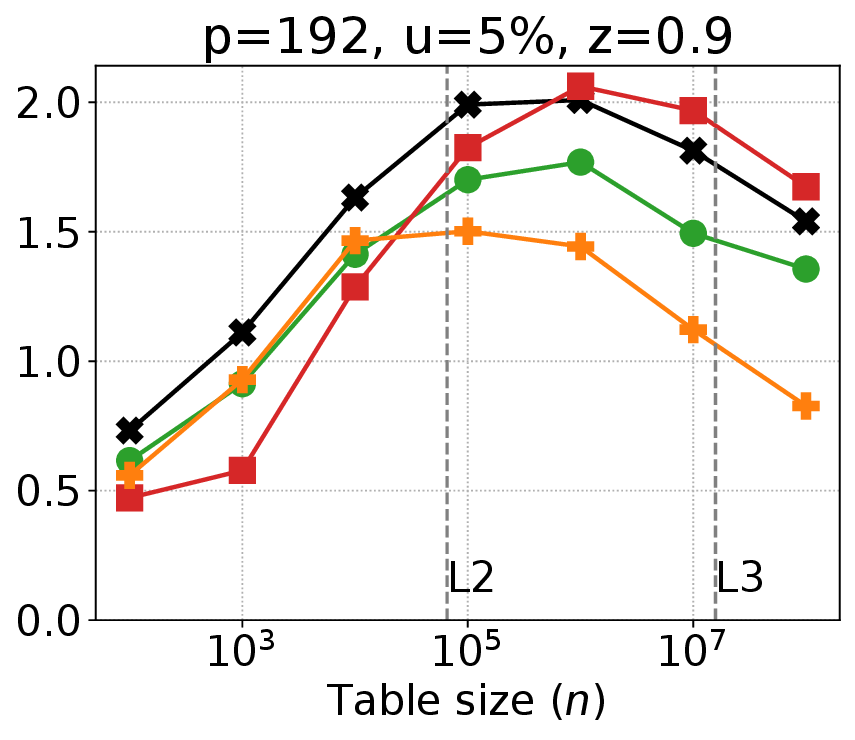}
    \includegraphics[height=3.4cm]{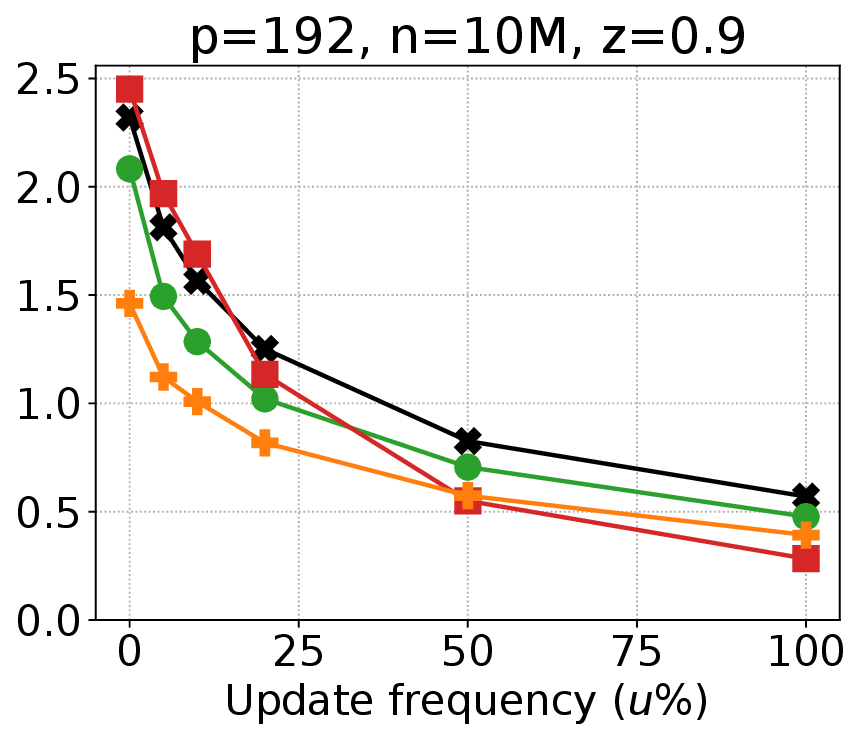}

    \caption{Throughput in billions of operations per second for our \ourHashtable{} hashtable with different big atomic implementations, varying thread counts $(p)$, update frequencies $(u)$, contention parameters $(z)$ and table sizes $(n)$.\vspace{-.1in}}
    \label{fig:c7i-hash-list}
\end{figure}

{
\myparagraph{Comparison to open-source hashtables}\label{sec:otherhash}
We compare \ourHashtable{} to other state-of-the-art open-source hash tables, including:
\begin{itemize}[leftmargin=*]
    \item \lstinline{libcuckoo}~\cite{cuckoo14,libcuckoo},
    \item Facebook's \lstinline{folly::ConcurrentHashMap}~\cite{follyF14} 
(\emph{Folly}),
    \item Intel's (\emph{TBB}) \lstinline{tbb::concurrent_hash_map}~\cite{tbb07},
    \item \lstinline{seq::concurrent_map}~\cite{seqhash},
    \item \lstinline{parallel_hashmap}~\cite{parhashmap}, which is based on Google's Swiss table~\cite{swisstable}, 
\end{itemize}

\noindent We use the implementation of \ourHashtable{} based on \came{} since it consistently performed the best.
Our results are depicted in Figure~\ref{fig:hashtable-comparison} comparing \ourHashtable{}  with the other tables. We see that \ourHashtable{} consistently performs the best across almost all parameter choices. Of the open-source tables, there is no implementation that consistently performs the best; the best implementation in one workload can be the worst in another. At low contention, \lstinline{seq::concurrent_map} scales almost as well as \ourHashtable{} until high thread counts.  However, at high contention, it becomes the second-worst of all the tables, beating only \lstinline{libcuckoo}, which is the second-best for low contention and low thread counts, but the worst otherwise. At low contention, Folly exhibits the worst performance, but at high contention, the best of the open-source tables. As a general trend, it is noteworthy that every open-source table significantly degrades in performance as contention increases, while \ourHashtable{} remains relatively flat.


    \vspace{-0.03in}
\section{Conclusion and Final Remarks}\label{sec:conclusion}


In this paper, we have designed big atomics that simultaneously achieve theoretical and practical efficiency.
The \came algorithm, in particular, simultaneously achieves strong theoretical properties such as lock-freedom and near-optimal memory efficiency, while exhibiting top performance across parameter regimes and remaining resilient to oversubscription.
The utility of these efficient big atomics is further bolstered by the performance of \ourHashtable, which consistently outperforms the fastest previously known concurrent hashtable implementations.

Our work gives rise to several interesting open questions.
(1)~Is there a fast, wait-free big atomic implementation that can rival the space efficiency of \came?
(2)~Are fully direct implementations of big atomics provably impossible?
(3)~What applications beyond hashtables benefit significantly from big atomics?
(4)~Are there efficient implementations of big atomics for other RMW operations, such as fetch-and-store and fetch-and-add?

\begin{acks}
    We thank the anonymous referees for their comments. This work was supported by the National Science Foundation grant CCF-2119352, a Dartmouth Start-up Grant, a Walter and Constance Burke Research Initiation Award, and a 
    gift from Jane Street.
\end{acks}

    \clearpage
    \bibliographystyle{ACM-Reference-Format}
    \bibliography{strings,biblio}

@inproceedings{Singh2025Publish,
  title     = {Publish on Ping: A Better Way to Publish Reservations in Memory Reclamation for Concurrent Data Structures},
  author    = {Singh, Amandeep and Brown, Trevor},
  booktitle = ppopp,
  pages     = {128--141},
  year      = {2025},
  doi       = {10.1145/3710848.3710890},
  url       = {https://dl.acm.org/doi/10.1145/3710848.3710890}
}

@inproceedings{bragagnolo2016fast,
  author    = {Bragagnolo, Santiago and Harris, Tim and Guerraoui, Rachid},
  title     = {Fast and Robust Memory Reclamation for Concurrent Data Structures},
  booktitle = spaa,
  pages     = {133--142},
  year      = {2016},
  publisher = {ACM},
  doi       = {10.1145/2935764.2935790},
  url       = {https://doi.org}
}

@misc{parlayhash,
  title={ParlayHash},
  author={Guy E. Blelloch and Siddhartha Jayanti and Jakub Lacki},
  year = "2025",
  note={{\scriptsize\url{https://github.com/cmuparlay/parlayhash}}}
}

@article{richter2015seven,
  author    = {Richter, Stefan and Leis, Viktor and Neumann, Thomas and Kemper, Alfons},
  title     = {A Seven-Dimensional Analysis of Hashing Methods and its Implications on Query Processing},
  journal   = {Proceedings of the VLDB Endowment},
  volume    = {8},
  number    = {9},
  pages     = {962--973},
  year      = {2015},
  publisher = {VLDB Endowment},
  doi       = {10.14778/2777598.2777606}
}

@inproceedings{BenderKK21,
  author    = {Michael A. Bender and Bradley C. Kuszmaul and William Kuszmaul},
  title     = {Linear Probing Revisited: Tombstones Mark the Demise of Primary Clustering},
  booktitle = focs,
  pages     = {1171--1182},
  year      = {2021},
  organization = {IEEE},
  doi       = {10.1109/FOCS52894.2021.00115}
}

@inproceedings{originalllsc,
author = {J. M. Broughton and P. M. Farmwald and T. M. McWilliams},
title = {{S-1 Multiprocessor System}},
volume = {0341},
booktitle = {Real-Time Signal Processing V},
editor = {Joel Trimble},
organization = {International Society for Optics and Photonics},
publisher = {SPIE},
pages = {327 -- 332},
year = {1982},
}

@inproceedings{jayanti2003efficient,
  author     = {Jayanti, Prasad and Petrovic, Srdjan},
  title      = {Efficient and Practical Constructions of {LL/SC} Variables},
  booktitle = podc,
  year       = {2003},
  pages      = {285--294},
  publisher = {Association for Computing Machinery},
  address   = {New York, NY, USA},
  doi       = {10.1145/872035.872078},
  url       = {https://doi.org}
}

@inproceedings{Moir1997Practical,
  title={Practical implementations of non-blocking synchronization primitives},
  author={Moir, Mark},
  booktitle=podc,
  pages={219--228},
  year={1997},
  organization={ACM}
}

@inproceedings{sheffi2023era,
  author    = {Sheffi, Gali and Petrank, Erez},
  title     = {The ERA Theorem for Safe Memory Reclamation},
  booktitle = podc,
  pages     = {352--362},
  year      = {2023},
  publisher = {Association for Computing Machinery},
  address   = {New York, NY, USA},
  doi       = {10.1145/3583668.3594564},
  url       = {https://doi.org}
}

@inproceedings{anderson2022turning,
  author    = {Anderson, Daniel and Blelloch, Guy E. and Wei, Yuanhao},
  title     = {Turning Manual Concurrent Memory Reclamation into Automatic Reference Counting},
  booktitle = pldi,
  pages     = {61--75},
  year      = {2022},
  publisher = {Association for Computing Machinery},
  address   = {New York, NY, USA},
  doi       = {10.1145/3519939.3523730},
  url       = {https://doi.org}
}

@techreport{treiber1986systems,
  author      = {Treiber, R. Kent},
  title       = {Systems Programming: Coping with Parallelism},
  institution = {IBM Almaden Research Center},
  number      = {RJ 5118},
  year        = {1986},
  month       = {April}
}

@inproceedings{harris2002practical,
  author    = {Harris, Timothy L. and Fraser, Keir and Pratt, Keir},
  title     = {A Practical Multi-word Compare-and-Swap Operation},
  booktitle = disc,
  pages     = {265--279},
  year      = {2002},
  publisher = {Springer-Verlag},
  address   = {Berlin, Heidelberg},
  doi       = {10.1007/3-540-36108-1_18},
  url       = {https://doi.org}
}

@inproceedings{guerraoui2020efficient,
  author    = {Guerraoui, Rachid and Guirguis, Arvid and Komatovic, Mirjana},
  title     = {Efficient Multi-Word Compare and Swap},
  booktitle = disc,
  volume    = {179},
  pages     = {4:1--4:17},
  year      = {2020},
  publisher = {Schloss Dagstuhl -- Leibniz-Zentrum f{\"u}r Informatik},
  doi       = {10.4230/LIPIcs.DISC.2020.4},
  url       = {https://doi.org}
}

@inproceedings{fences2011,
  author    = {Ladan-Mozes, Edya and Lee, I-Ting Angelina and Vyukov, Dmitry},
  title     = {Location-Based Memory Fences},
  booktitle = spaa,
  year      = {2011},
  pages     = {75--84},
  publisher = {Association for Computing Machinery},
  address   = {New York, NY, USA},
  doi       = {10.1145/1989493.1989503},
  url       = {https://dl.acm.org/doi/10.1145/1989493.1989503}
}

@misc{blelloch2020concurrent,
      title={Concurrent Fixed-Size Allocation and Free in Constant Time}, 
      author={Guy E. Blelloch and Yuanhao Wei},
      year={2020},
      eprint={2008.04296},
      archivePrefix={arXiv},
      primaryClass={cs.DC}
}

@misc{hpcpp26,
    title={Hazard Pointers for C++26},
    author={Maged M. Michael and  Michael Wong and Paul McKenney and Andrew Hunter and Daisy S. Hollman and JF Bastien and Hans
Boehm and David Goldblatt and Frank Birbacher and Mathias Stearn},
    year={2023},
    note={{\scriptsize\url{https://www.open-std.org/jtc1/sc22/wg21/docs/papers/2023/p2530r3.pdf}}}
}

@phdthesis{Sullivan17,
title = {Low-level Concurrent Programming Using the
Relaxed Memory Calculus},
author = {Michael J. Sullivan},
school = {Carnegie Mellon University},
year = 2017,
note = {CMU-CS-17-126}}

@inproceedings{Boehm12,
author = {Hans-J. Boehm},
title = {Can seqlocks get along with programming language memory models?},
booktitle = {ACM SIGPLAN Workshop on Memory Systems Performance and
Correctness},
year = 2012}

@misc{Hemminger02,
author = {Stephen Hemminger},
title = {Fast reader/writer lock for gettimeofday 2.5.30. Linux kernel mailing
list},
year = 2012,
howpublished = {\url{https://lwn.net/Articles/7388/}}
}

@inproceedings{Lameter05,
author = {C. Lameter},
title = {Effective synchronization on Linux/NUMA systems},
booktitle = {Proc. of the Gelato Federation Meeting},
year = 2005,
url = {http://lameter.com/gelato2005.pdf},
}

@article{Herlihy93,
  author       = {Maurice Herlihy},
  title        = {A Methodology for Implementing Highly Concurrent Objects},
  journal      = {{ACM} Trans. Program. Lang. Syst.},
  volume       = {15},
  number       = {5},
  pages        = {745--770},
  year         = {1993},
}

@misc{abseil,
author = "Google",
title = "Abseil: C++ Library",
url = "https://abseil.io/",
year = 2023
}

@inproceedings{YCSB,
	author = {Cooper, Brian F. and Silberstein, Adam and Tam, Erwin and Ramakrishnan, Raghu and Sears, Russell},
	title = {Benchmarking Cloud Serving Systems with {YCSB}},
	year = {2010},
	isbn = {9781450300360},
	url = {https://doi.org/10.1145/1807128.1807152},
	doi = {10.1145/1807128.1807152},
	booktitle = {Proc.~1st ACM Symposium on Cloud Computing},
	pages = {143--154},
	numpages = {12},
}

@techreport{fraser2004practical,
  title={Practical lock-freedom},
  author={Fraser, Keir},
  year={2004},
  institution={University of Cambridge, Computer Laboratory}
}

@InProceedings{BlellochWei20,
	author =	{Guy E. Blelloch and Yuanhao Wei},
	title =	{{LL/SC and Atomic Copy: Constant Time, Space Efficient Implementations Using Only Pointer-Width CAS}},
	booktitle =	{International Symposium on Distributed Computing (DISC)},
	year =	{2020},
}

@inproceedings{harris2001pragmatic,
	title={A pragmatic implementation of non-blocking linked-lists},
	author={Harris, Timothy L},
	booktitle=disc,
	pages={300--314},
	year={2001},
	organization={Springer}
}

@inproceedings{Natarajan14,
author = {Natarajan, Aravind and Mittal, Neeraj},
title = {Fast Concurrent Lock-Free Binary Search Trees},
year = {2014},
booktitle={ppopp},
}

@inproceedings{HLMS03,
author = {Herlihy, Maurice and Luchangco, Victor and Moir, Mark and Scherer, William N.},
title = {Software Transactional Memory for Dynamic-Sized Data Structures},
year = {2003},
booktitle = {podc},
}

@inproceedings{WeiBBFRS21,
  author    = {Yuanhao Wei and
               Naama Ben{-}David and
               Guy E. Blelloch and
               Panagiota Fatourou and
               Eric Ruppert and
               Yihan Sun},
  title     = {Constant-time snapshots with applications to concurrent data structures},
  booktitle = ppopp,
  pages     = {31--46},
  year      = {2021},
}

@article{Wu17,
 author = {Wu, Yingjun and Arulraj, Joy and Lin, Jiexi and Xian, Ran and Pavlo, Andrew},
 title = {An Empirical Evaluation of In-memory Multi-version Concurrency Control},
 journal = vldb,
 issue_date = {March 2017},
 volume = {10},
 number = {7},
 month = mar,
 year = {2017},
 pages = {781--792},
}

@techreport{Reed78,
author = "D. Reed",
title = "Naming and synchronization in a decentralized computer system",
institution = "EECS Dept., MIT",
number = "LCS/TR-205",
month = sep,
year = 1978}

@article{BG83,
 author = {Bernstein, Philip A. and Goodman, Nathan},
 title = {Multiversion Concurrency Control - Theory and Algorithms},
 journal = {ACM Trans. Database Syst.},
 issue_date = {Dec. 1983},
 volume = {8},
 number = {4},
 month = dec,
 year = {1983},
 issn = {0362-5915},
 pages = {465--483},
 numpages = {19},
 url = {http://doi.acm.org/10.1145/319996.319998},
 doi = {10.1145/319996.319998},
 acmid = {319998},
 publisher = {ACM},
 address = {New York, NY, USA},
}

@inproceedings{neumann2015fast,
	title={Fast Serializable Multi-version Concurrency Control for Main-Memory Database Systems},
	author={Neumann, Thomas and M{\"u}hlbauer, Tobias and Kemper, Alfons},
	booktitle=sigmod,
	pages={677--689},
	year={2015},
}

@inproceedings{LKA17,
author = {Lim, Hyeontaek and Kaminsky, Michael and Andersen, David G.},
title = {Cicada: Dependably Fast Multi-Core In-Memory Transactions},
year = {2017},
doi = {10.1145/3035918.3064015},
booktitle = sigmod,
pages = {21–35}
}

@article{Postgres12,
 author = {Ports, Dan R. K. and Grittner, Kevin},
 title = {Serializable Snapshot Isolation in {PostgreSQL}},
 journal = {Proc.~of the VLDB Endowment},
 issue_date = {August 2012},
 volume = {5},
 number = {12},
 month = aug,
 year = {2012},
 issn = {2150-8097},
 pages = {1850--1861},
 numpages = {12},
 url = {http://dx.doi.org/10.14778/2367502.2367523},
 doi = {10.14778/2367502.2367523},
 acmid = {2367523},
 publisher = {VLDB Endowment},
}

@inproceedings{SQL13,
 author = {Diaconu, Cristian and Freedman, Craig and Ismert, Erik and Larson, Per-Ake and Mittal, Pravin and Stonecipher, Ryan and Verma, Nitin and Zwilling, Mike},
 title = {Hekaton: {SQL} Server's Memory-optimized OLTP Engine},
 booktitle = sigmod,
 year = {2013},
 pages = {1243--1254},
 doi = {10.1145/2463676.2463710},
}

@ARTICLE{hazard04,
 author={Michael, M.M.},
  journal={IEEE Transactions on Parallel and Distributed Systems}, 
  title={Hazard pointers: safe memory reclamation for lock-free objects}, 
  year={2004},
  volume={15},
  number={6},
  pages={491-504},
  doi={10.1109/TPDS.2004.8}}

@inproceedings{cuckoo14,
author = {Li, Xiaozhou and Andersen, David G. and Kaminsky, Michael and Freedman, Michael J.},
title = {Algorithmic Improvements for Fast Concurrent Cuckoo Hashing},
year = {2014},
isbn = {9781450327046},
publisher = {Association for Computing Machinery},
address = {New York, NY, USA},
url = {https://doi.org/10.1145/2592798.2592820},
doi = {10.1145/2592798.2592820},
booktitle = eurosys,
articleno = {27},
numpages = {14},
location = {Amsterdam, The Netherlands},
series = {EuroSys '14}
}

@misc{onetbb,
author = "Intel",
year = 2023,
title = "Intel OneTBB",
url = "https://github.com/oneapi-src/oneTBB"}

@misc{libcuckoo,
author = "Manu Goyal and Bin Fan and Xiaozhou Li and David G. Andersen and and Michael Kaminsky",
year = 2022,
title = "libcuckoo",
url = "https://github.com/efficient/libcuckoo"}

@misc{follyF14,
author = " Nathan Bronson and Xiao Shi",
title = "Open-sourcing F14 for faster, more memory-efficient hash tables",
year = 2019,
url = {https://engineering.fb.com/2019/04/25/developer-tools/f14/}
}

@misc{swisstable,
author = {Sam Benzaquen and Alkis Evlogimenos and Matt Kulukundis and Roman Perepelitsa},
title = {Swiss Table Design Notes},
year = 2017,
url = {https://abseil.io/about/design/swisstables}
}

@misc{parhashmap,
author = "Gregory Popovitch",
title = "The Parallel Hashmap",
year = 2025,
url = {https://github.com/greg7mdp/parallel-hashmap}
}

@misc{seqhash,
author = "Victor Moncada",
title = "Seq Concurrent Map",
year = 2025,
url = {https://github.com/Thermadiag/seq/blob/main/docs/concurrent_map.md}}

@misc{BoostHash,
author = "Joaquín M López Muñoz",
title = "Inside boost::concurrent\_flat\_map",
year = 2023,
url = {http://bannalia.blogspot.com/2023/07/inside-boostconcurrentflatmap.html}
}

@book{tbb07,
author = {Reinders, James},
title = {Intel Threading Building Blocks},
year = {2007},
isbn = {9780596514808},
publisher = {O'Reilly \& Associates, Inc.},
address = {USA},
edition = {First}
}

@book{HerlihyShavitBook,
author = {Herlihy, Maurice and Shavit, Nir},
title = {The Art of Multiprocessor Programming, Revised Reprint},
year = {2012},
isbn = {9780123973375},
publisher = {Morgan Kaufmann Publishers Inc.},
address = {San Francisco, CA, USA},
edition = {1st}
}

@misc{htmdead,
author = "Gareth Halfacree",
year = 2021,
title = {Intel sticks another nail in the coffin of TSX with feature-disabling microcode update},
publisher = {The Register},
date = {June 29},
url = {https://www.theregister.com/2021/06/29/intel_tsx_disabled/}
}

@article{htm,
author = {Herlihy, Maurice and Moss, J. Eliot B.},
title = {Transactional memory: architectural support for lock-free data structures},
year = {1993},
journal = {SIGARCH Comput. Archit. News},
publisher = {Association for Computing Machinery},
address = {New York, NY, USA},
volume = {21},
number = {2}
}

@inproceedings{michael04,
author = {Maged M Michael},
title = {Practical lock-free and wait-free LL/SC/VL implementations using 64-bit
CAS},
booktitle = disc,
year = 2004}

@inproceedings{JayantiP05,
    author = {Prasad Jayanti and Srdjan Petrovic},
    title = {Efficiently implementing a large number of LL/SC objects},
    booktitle = opodis,
    year = 2005
}

@inproceedings{JayantiJayantiVersionedSnapshot,
author = {Jayanti, Prasad and Jayanti, Siddhartha Visveswara},
title = {A Shared Archive of Snapshots},
year = {2025},
publisher = {Association for Computing Machinery},
address = {New York, NY, USA},
booktitle = {Proceedings of the ACM Symposium on Principles of Distributed Computing},
pages = {466–476},
numpages = {11},
series = {PODC '25}
}

@article{AndersonMoir99,
  author={Anderson, J.H. and Moir, M.},
  journal={IEEE Transactions on Parallel and Distributed Systems}, 
  title={Universal constructions for large objects}, 
  year={1999},
  volume={10},
  number={12}
}

@article{JT21,
  author       = {Siddhartha V. Jayanti and
                  Robert E. Tarjan},
  title        = {Concurrent disjoint set union},
  journal      = {Distributed Comput.},
  volume       = {34},
  number       = {6},
  pages        = {413--436},
  year         = {2021},
  url          = {https://doi.org/10.1007/s00446-020-00388-x},
  doi          = {10.1007/S00446-020-00388-X},
  bibsource    = {dblp computer science bibliography, https://dblp.org}
}

@inproceedings{JTB19,
author = {Jayanti, Siddhartha and Tarjan, Robert E. and Boix-Adser\`{a}, Enric},
title = {Randomized Concurrent Set Union and Generalized Wake-Up},
year = {2019},
publisher = {Association for Computing Machinery},
address = {New York, NY, USA},
url = {https://doi.org/10.1145/3293611.3331593},
doi = {10.1145/3293611.3331593},
booktitle = {podc},
}

@phdthesis{JPhD23,
  author       = {Siddhartha Visveswara Jayanti},
  title        = {Simple, Fast, Scalable, and Reliable Multiprocessor Algorithms},
  school       = {MIT, {USA}},
  year         = {2023},
  url          = {https://hdl.handle.net/1721.1/150219},
  bibsource    = {dblp computer science bibliography, https://dblp.org}
}

@InProceedings{BCW24,
  author =	{Bashari, Benyamin and Chan, David Yu Cheng and Woelfel, Philipp},
  title =	{{A Fully Concurrent Adaptive Snapshot Object for RMWable Shared-Memory}},
  booktitle =	disc,
  year =	{2024},
  volume =	{319},
  doi =		{10.4230/LIPIcs.DISC.2024.7},  
}

@InProceedings{EW13,
author="Ellen, Faith
and Woelfel, Philipp",
editor="Afek, Yehuda",
title="An Optimal Implementation of Fetch-and-Increment",
booktitle="Distributed Computing",
year="2013",
publisher="Springer Berlin Heidelberg",
pages="284--298",
}

@inproceedings{NW24,
author = {Naderi-Semiromi, Fatemeh and Woelfel, Philipp},
title = {Strongly Linearizable LL/SC from CAS},
year = {2024},
doi = {10.1145/3662158.3662800},
booktitle = podc
}

@inproceedings{BGW24,
author = {Bencivenga, Dante and Giakkoupis, George and Woelfel, Philipp},
title = {Faster Randomized Repeated Choice and DCAS},
year = {2024},
booktitle = podc
}

@inproceedings{BJW23,
author = {Bashari, Benyamin and Jamadi, Ali and Woelfel, Philipp},
title = {Efficient Bounded Timestamping from Standard Synchronization Primitives},
year = {2023},
doi = {10.1145/3583668.3594601},
booktitle = podc
}

@inproceedings{bravo19,
 author = {Dice, Dave and Kogan, Alex}, 
 title = {BRAVO: biased locking for reader-writer locks}, 
 year = {2019}, isbn = {9781939133038}, 
 publisher = {USENIX Association}, address = {USA}, booktitle = {Proceedings of the 2019 USENIX Conference on Usenix Annual Technical Conference}, pages = {315–328}, numpages = {14}, location = {Renton, WA, USA}, series = {USENIX ATC '19} }

@online{abseil_tip_197,
  author       = {{Abseil Team}},
  title        = {Tip of the Week \#197: Reader Locks Should Be Rare},
  month        = aug,
  year         = {2021},
  url          = {https://abseil.io},
  urldate      = {2026-05-27},
  organization = {Google}
}

@inproceedings{AndersonBKJ26,
  author       = {Daniel Anderson and
                  Guy E. Blelloch and
                  Zachary Kent and
                  Siddhartha Jayanti},
  title        = {Big Atomics: Non-Blocking Algorithms with a Direct Fast Path},
  booktitle    = spaa, 
  publisher    = {{ACM}},
  year         = {2026},
  url          = {https://doi.org/10.1145/3816782.3819220},
  doi          = {10.1145/3816782.3819220},
}

@string{spaa = "{ACM} Symposium on Parallelism in Algorithms
	and Architectures (SPAA)"}

@string{ppopp = "{ACM} Symposium on Principles and
	Practice of Parallel Programming (PPOPP)"}

@string{focs = "{IEEE} Symposium on Foundations of Computer Science (FOCS)"}

@string{pldi = "ACM Conference on Programming Language Design and Implementation (PLDI)"}

@string{podc = "{ACM} Symposium on Principles of Distributed Computing (PODC)"}

@string{disc = "International Symposium on Distributed Computing (DISC)"}

@string{fast = "{USENIX} Conference on File and Storage Technologies (FAST)"}

@string{vldb = "Proceedings of the VLDB Endowment (PVLDB)"}

@string{sigmod = "ACM SIGMOD International Conference on Management of Data (SIGMOD)"}

@string{eurosys = "ACM European Conference on Computer Systems (EuroSys)"}

@string{opodis = "Conf. on Principles of Distributed Systems (OPODIS)"}
    
    \clearpage

        
\end{document}